\documentclass[11pt]{article}
\usepackage[margin=1.5in]{geometry}
\usepackage{natbib,amsmath,amssymb,amsthm,graphicx,setspace,paralist,booktabs,rotating,
	subcaption,float}
\usepackage{times}
\usepackage{xcolor}
\usepackage{multirow}
\usepackage{bm, bbm}
\usepackage{url}
\usepackage{hyperref}
\usepackage{enumitem}

\usepackage{mathrsfs}
\usepackage{natbib}
\usepackage[]{algorithm2e}

\newcommand{\cid}{\overset{d}{\to}}
\newcommand{\indep}{\raisebox{0.05em}{\rotatebox[origin=c]{90}{$\models$}}}

\def\argmax{\mathop{\rm arg\, max}}

\def\R{\mathbb{R}}
\def\E{\mathbb{E}}
\def\P{\mathbb{P}}

\def\eps{\epsilon}
\def\Sig{\Sigma}
\def\sig{\sigma}

\def\lam{\lambda}
\def\gam{\gamma}

\newtheorem{condition}{Condition}[section]
\newtheorem{theorem}{Theorem}[section]
\newtheorem{corollary}{Corollary}[section]
\newtheorem{proposition}{Proposition}[section]
\newtheorem{lemma}{Lemma}[section]
\newtheorem{remark}{Remark}[section]

\def\calV{\mathcal{V}}

\def \dg{\triangledown g}
\def \hg{\triangle g}

\def \hprob{1-n^{-c}}

\usepackage{tikz}
\usepackage{verbatim}
\usepackage[listings,theorems]{tcolorbox}
\usetikzlibrary{chains,fit,arrows,shapes,snakes}
\usetikzlibrary{shapes.misc}
\usetikzlibrary{shapes.multipart}
\usetikzlibrary{positioning}
\usetikzlibrary{decorations.pathreplacing,angles,quotes}
\usepackage{smartdiagram}

\usepackage[affil-it]{authblk}
\usepackage{xr}
\begin{document}

\begin{titlepage}

  \title{\bf Causal Inference for Nonlinear Outcome Models with Possibly Invalid Instrumental Variables}
  \author[1]{Sai Li}
  \author[2,*] {Zijian Guo}
 
 \affil[1]{Institute of Statistics and Big Data, Renmin University of China, China.}
\affil[2]{Department of Statistics, Rutgers University, Piscataway, NJ 08854. }
\affil[*]{Corresponding author: Zijian Guo, zijguo@stat.rutgers.edu}
  \maketitle

\bigskip

\begin{abstract}
Instrumental variable methods are widely used for inferring the causal effect in the presence of unmeasured confounders. Existing instrumental variable methods for nonlinear outcome models require stringent identifiability conditions. This paper considers a flexible semi-parametric potential outcome model that allows for possibly invalid instruments. We propose new identifiability conditions to identify the causal parameters when the majority of the instrumental variables are valid. We devise a novel inference procedure for a new average structural function and the conditional average treatment effect. We establish the asymptotic normality of the proposed estimators and construct confidence intervals for the causal estimands by bootstrap. The proposed method is demonstrated in large-scale simulation studies and is applied to infer the effect of income on house ownership. 
\end{abstract}

\vfill
\end{titlepage}
\newpage

\section{Introduction}
\label{sec1}
Unmeasured confounders are significant concerns for inferring causal effects from observational studies. The instrumental variable (IV) approach is the state-of-the-art method for estimating the causal effects in the presence of unmeasured confounders \citep[e.g.]{wooldridge2010econometric}. 
The success of IV-based methods requires the IVs to
satisfy three core conditions: conditioning on the observed covariates,\\
 (A1) the IVs are associated with the exposure;\\
 (A2) the IVs have no direct effects on the outcome;\\
 (A3) the IVs are independent of unmeasured confounders. 
 
A significant challenge of applying IV-based methods is to identify valid instruments simultaneously satisfying (A1)-(A3) \citep[e.g.]{murray2006avoiding,conley2012plausibly}. Since (A2) and (A3) cannot be tested in a data-dependent way, it requires strong domain knowledge to identify valid IVs. In many applications, the relationship between the outcome and treatment is nonlinear. For example, when the outcome is binary, the treatment affects the outcome in a nonlinear way. Inference for the causal effect under the nonlinear outcome models is more challenging than that for the linear model. There is a pressing need to develop accurate causal inference methods in nonlinear outcome models with possibly invalid instruments.

\subsection{Existing works}

Semi-parametric approaches are widely used for causal inference with nonlinear outcome models. \citet{BP04} and \citet{Rothe09} considered the double-index models for binary outcomes by assuming a valid control function and made inferences for the model parameters. In practice, it may be more interesting to make inferences for the conditional average treatment effect (CATE) and the average structural function (ASF) \citep{blundell2003endogeneity}, which are measured on the same scale as the outcome variable. \citet{hahn2013asymptotic} studied the influence function of partial means estimators with semi-parametric outcome models and generated covariates. A typical example for generated covariates is a valid control function. By assuming a valid control function, the ASF can be estimated via the partial means method \citep{Newey94} and its asymptotic properties have been studied in \citet{mammen2012nonparametric}. However, the violations of the valid control function assumption are barely studied in the semi-parametric outcome models.

For causal inference with binary responses, parametric models such as probit and logistic outcome models \citep{Rivers88, Vans11} were studied with valid control functions and IVs. In recent work, \citet{carlson2021relaxing} considered causal inference in a probit outcome model and a parametric treatment model. They relaxed the valid control function condition in \citet{Rivers88} and allowed the unobserved confounders to depend on the IVs through a specific form. However, these parametric models required specific distributions of the unmeasured confounders, which can be misspecified in practical applications. The mixed-logistic model, stated in the following equation (\ref{eq: mixed logit}), is commonly used in observational studies \citep{CW12}. However, even assuming a valid IV, the two-stage method is known to be biased for the mixed-logistic model \citep{CD11} .

Under linear outcome models, some recent progress has been made in inferring the causal effects with possibly invalid IVs. With continuous outcome and exposure models, \citet{kolesar2015identification} and \citet{Bowden15} proposed methods in the setting that all candidate IVs can be invalid, but the IV strength and the direct effect on the outcome are nearly orthogonal. \citet{Bowden16}, \citet{Kang16}, and \citet{Wind19} proposed consistent estimators of causal effects assuming at least 50\% of the IVs are valid. \citet{Bowden17} and \citet{TSTH} considered linear outcome models assuming that the most common causal effect estimate is a consistent estimate of the true causal effect. Under this assumption, \citet{TSTH} constructed confidence intervals for the treatment effect, and \citet{Wind19b} further developed the inference procedure built on \citet{TSTH}. However, these methods are developed under the linear models and do not apply to nonlinear outcome models. 

In the GMM framework, \citet{liao2013adaptive,cheng2015select,caner2018adaptive} assumed the prior knowledge of a set of valid moments and considered selecting valid moments among another set of moments. \citet{ditraglia2016using} leveraged invalid moments to reduce the mean squared error with the prior knowledge of a set of valid moments. In contrast, the current paper does not assume any prior knowledge of the validity of any given instrument.

\subsection{Our results and contributions}
We propose a robust causal inference method for semi-parametric outcome models with possibly invalid IVs. We relax the valid control function condition (Condition \ref{cond: valid IC} in Section \ref{sec2-1}) by allowing for a violation in a semi-parametric form. It generalizes the linear violation forms considered in linear outcome models \citep[e.g.]{Kang16, TSTH}.

We impose the dimension reduction condition (Condition \ref{cond: dim red} in Section \ref{sec: identify}) and the majority rule (Condition \ref{cond: majority} in Section \ref{sec: identify}) as new identification conditions for semi-parametric outcome models with possibly invalid IVs. These new identifiability conditions significantly weaken the commonly used Condition \ref{cond: valid IC}. We show that the new identifiability conditions are sufficient to identify a new ASF and CATE. 

We propose a three-step inference procedure for CATE in \underline{\textbf{S}}emi-\underline{\textbf{p}}arametric \underline{\textbf{o}}u\underline{\textbf{t}}come models with possibly invalid \underline{\textbf{IV}}s, termed as SpotIV. First, we estimate the reduced-form parameters based on the existing dimension reduction methods. Second, we apply the median estimator to estimate the model parameters by leveraging that more than $50\%$ of candidate IVs are valid. Third, we develop a partial mean estimator for the causal estimand of interest. We further develop a self-checking method to partially test whether the majority rule is satisfied.

We establish  the asymptotic normality of our proposed estimator and construct confidence intervals for the causal estimands by bootstrap. We demonstrate our proposed method in simulations and apply the method to infer the causal effects of household income on whether a family owns a house based on the China Family Panel Studies (CFPS).

\subsection{Organization of the rest of the paper}

The rest of this paper is organized as follows. 
In Section \ref{sec2}, we introduce the model and identifiability conditions. In Section \ref{sec3}, the SpotIV estimator is proposed to make inference for ASF and CATE. Section \ref{sec4} provides theoretical guarantees for the proposed method. In Section \ref{sec-simu}, we investigate the empirical performance of the SpotIV estimator. In Section \ref{sec-data}, the SpotIV estimator is applied to infer the income's causal effects on house ownership.

\section{Nonlinear outcome models and identifiability conditions}
\label{sec2}
For the $i$-th subject, let $y_i\in \R$ denote the observed outcome, $d_i\in \R$ denote the exposure or the treatment, $z_i\in\R^{p_z}$ denote a set of candidate IVs, and $x_i\in\R^{p_x}$ denote baseline covariates.  Define $p:=p_z+p_x$ and  we use $w_i:=(z_i^{\intercal}, x_i^{\intercal})^{\intercal} \in \R^{p}$ to denote all the measured covariates, including  candidate IVs and baseline covariates. We assume that the data $\{y_{i},d_{i},w_i\}_{1\leq i\leq n}$ are generated in \textit{i.i.d.} fashions. Let $u_i$ denote the unmeasured confounder which can be associated with both exposure and outcome variables.

We define causal effects using the potential outcome framework \citep{Neyman23, Rubin74}.
Let $y_i^{(d)}\in\R$  be the potential outcome if the $i$-th individual were to have exposure $d$.  We consider the following nonlinear potential outcome model
\begin{equation}
\label{eq: potential}
     \E[y^{(d)}_i|w_i=w, u_i=u]:=q\left(d\beta+w^{\intercal}\kappa, u\right),
\end{equation}
where $q: \R^2 \rightarrow \R$ is a link function, $\beta\in\R$ is the coefficient corresponding to the exposure, and $\kappa=(\kappa_z^{\intercal},\kappa_x^{\intercal})^{\intercal}\in\R^{p}$ is the coefficient vector corresponding to both the IVs and baseline covariates. The function $q$ can be either known or unknown. 
Model \eqref{eq: potential} includes a broad class of nonlinear potential outcome models for both continuous and binary outcomes. For binary outcomes, if $q(a,b)=1/(1+\exp(-a-b))$, then \eqref{eq: potential} is the mixed-logistic model \citep{ Vans11}; if $q(a,b)=\mathbbm{1}(a+b>0)$ and $u_i$ is normal with mean zero, then \eqref{eq: potential} is the probit model \citep{Rivers88}. 

We assume that $y_i^{(d)} \; \indep \; d_i \mid (w_i^{\intercal}, u_i)$. This condition is mild as we can always identify an unmeasured variable $u_i$ such that $y_i^{(d)}$ and $d_i$ are conditionally independent. This is much weaker than the (strong) ignorability condition $y_i^{(d)} \; \indep \; d_i \mid w_i$ \citep{RB83}.
Under the condition $y_i^{(d)} \; \indep \; d_i \mid (w_i^{\intercal}, u_i)$ and the consistency assumption \citep[e.g.]{imbens2015causal}, we connect the conditional outcome model and the potential outcome model \eqref{eq: potential} as
\begin{equation}
\E[y_i|d_i=d,w_i=w,u_i=u]= \E[y^{(d)}_i|d_i=d,w_i=w,u_i=u]=\E[y^{(d)}_i|w_i=w,u_i=u].
\label{eq: connection}
\end{equation}
Consequently, the potential outcome model \eqref{eq: potential} leads to the following outcome model, \begin{equation}
  \E[y_i|d_i=d,w_i=w,u_i=u]=
q\left(d \beta+w^{\intercal}\kappa,u\right).
\label{eq-y}
\end{equation}
For the exposure $d_i$, we consider a linear working model
\begin{equation}
\label{eq-d}
d_i=w_i^{\intercal}\gam+v_i~~\text{with}~~\gam=(\E[w_iw_i^{\intercal}])^{-1}\E[w_id_i]~~\text{and}~~\E[w_iv_i]=0.
\end{equation}

When there exist unmeasured confounders, $u_i$ and $v_i$ are correlated, and the exposure is endogenous; that is, $d_i$ is associated with the unmeasured confounder $u_i$ even after conditioning on the measured variables $w_i$; see Figure \ref{fig: robust IV} for an illustration.


\subsection{Review of the control function approach with valid IVs}
\label{sec2-1}
We first review the control function approach with valid IVs, which is widely adopted for causal inference with nonlinear outcome models \citep{Rivers88, BP04, Rothe09, petrin2010control, CD11, wooldridge2015control, guo2016control}. The key idea is to treat the residual $v_i$ of the exposure model \eqref{eq-d} as a proxy for the unmeasured confounders. Then $v_i$ is included in the outcome model as an adjustment for the unmeasured confounder. The success of the control function method relies on the following identifiability condition \citep{BP04, Rothe09}.

 \begin{condition} [Control function with valid IVs]
 \label{cond: valid IC}
\textit{ The IV strength $\gam_z$ in \eqref{eq-d} satisfies $\|\gam_{z}\|_2\geq \tau_0>0$ for a positive constant $\tau_0$. The direct effect $\kappa_z$ in \eqref{eq-y} is equal to zero. The conditional density $f_u(u_i|w_i,v_i)$ satisfies 
\begin{equation}
\label{cond0-u}
f_u(u_i|w_i,v_i)=f_u(u_i|v_i).
\end{equation}}
\end{condition}
 The condition $\|\gam_{z}\|_2\geq \tau_0>0$ assumes strong associations between the IVs and the exposure variable, which corresponds to the IV assumption (A1). 
The condition $\kappa_{z}=0$ assumes that the IVs do not directly affect the outcome, which corresponds to (A2). Equation (\ref{cond0-u}) assumes that conditioning on the control variable $v_i$, the unmeasured confounder $u_i$ is independent of the measured covariates $w_i$. Note that \eqref{cond0-u} automatically holds if the errors $(u_i,v_i)$ are independent of $w_i.$ The assumption in \eqref{cond0-u} can be viewed as a version of (A3) for nonlinear outcome models. If $v_i$ is independent of $w_i$, then condition \eqref{cond0-u} implies (A3). However, such a connection may not hold in general. Under Condition \ref{cond: valid IC}, the outcome model (\ref{eq-y}) implies
\begin{equation}
\label{Valid-IC}
\E[y_i|d_i,w_i,v_i]=\int q(d_i\beta+w_i^{\intercal}\kappa,u_i)f_u(u_i|v_i)du_i:=g_0\left(d_i\beta+x_{i}^{\intercal}\kappa_x,v_i\right),
\end{equation}
where the function $g_0: \R^2\rightarrow \R$ is induced by the functions $q(\cdot,\cdot)$ and $f_u(\cdot)$. Inference for parameters $\beta$ and $\kappa_x$ in (\ref{Valid-IC}) has been studied in \citet{BP04} and \citet{Rothe09}. For nonlinear outcome models, \citet{blundell2003endogeneity} proposed the average structural function (ASF) as the targeted causal estimand, defined as
\begin{equation}
\label{eq-asf}
\textrm{ASF}(d,w)=\E_{u_i}[q(d\beta+w^{\intercal}\kappa,u_i)].
 \end{equation}
 Under Condition \ref{cond: valid IC}, it holds that
 \begin{equation}
 \label{eq-asf2}
 \textrm{ASF}(d,w)=\E[\E[q(d\beta+w^{\intercal}\kappa,u_i)|v_i]]=\E[g_0\left(d\beta+x^{\intercal}\kappa_x,v_i\right)],
 \end{equation}
 where the last step follows from \eqref{Valid-IC}.
 \citet{mammen2012nonparametric,hahn2013asymptotic} proposed partial mean estimators of ASF$(d,w)$ by leveraging the last expression in (\ref{eq-asf2}).

We shall highlight that Condition \ref{cond: valid IC} can be easily violated in practical applications. First, the assumption (\ref{cond0-u}) is unlikely to hold when $u_i$ involves omitted variables, which may be associated with measured covariates $w_i$.
As pointed out in \citet{BP04}, a control function satisfying Condition \ref{cond: valid IC} largely relies on including all the suspicious confounders in the model, which may be a strong assumption for practical applications. Second, \citet{carlson2021relaxing} pointed out that the assumption (\ref{cond0-u}) does not allow heteroskedasticity with respect to the distribution of $u_i$. For example, the model of $u_i|w_i, v_i$ specified in the following equation (\ref{carlson-cond}) violates Condition \ref{cond: valid IC}. Finally, the no direct effect assumption ($\kappa_z=0$) might be violated in practice. Indeed, both $\kappa_z=0$ and the assumption \eqref{cond0-u} cannot be tested in a data-dependent way. The above discussions strongly motivate us to propose much weaker identification conditions in the following subsection.

\subsection{New identifiability conditions}
\label{sec: identify}
We introduce new identifiability conditions, which weakens Condition \ref{cond: valid IC}. \begin{condition} [Dimension reduction]
 \label{cond: dim red}
 The direct effect $\kappa_z$ in (\ref{eq-y}) can be non-zero.
There exists some $\eta\in\R^{p\times p_{\eta}}$ for $0\leq p_{\eta}< p$ such that the conditional density $f_u(u_i|w_i,v_i)$ satisfies
\begin{equation}
\label{cond1-u}
f_u(u_i|w_i,v_i)=f_u(u_i|w_i^{\intercal}\eta, v_i).
\end{equation}
 \end{condition}
Without loss of generality, we assume the columns of $\eta$ are linearly independent and $p_{\eta}< p$.
In contrast to (\ref{cond0-u}), the model (\ref{cond1-u}) allows the unmeasured confounder $u_i$ to depend on the measured covariates $w_i$ through the linear transformations $w_i^{\intercal}\eta$, after conditioning on $v_i$. Condition \ref{cond: dim red} essentially requires a dimension reduction property of the conditional density $f_u(u_i|w_i,v_i)$. 
As illustrated in Figure \ref{fig: robust IV}, Condition \ref{cond: dim red} relaxes Condition \ref{cond: valid IC} by allowing $\kappa_z\neq \bm{0}$ and $\eta\neq \bm{0}$. 

\tikzstyle{format} = [draw, thin,minimum height=1.35em, minimum width=1.35cm]
\tikzstyle{format1} = [draw, thin,minimum height=1.35em, minimum width=2.35cm]
\tikzstyle{medium} = [draw, thin,minimum height=1.35em]

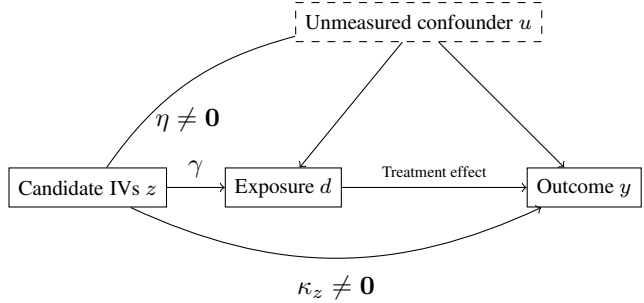
\begin{figure}[htb!]
\centering
\begin{tikzpicture}
  \node[format] at (9.5, 5) (treatment2) {\scriptsize Exposure $d$};
  \node[format, right of=treatment2, node distance=4cm] (outcome2) {\scriptsize Outcome $y$};
  \draw[->]  (treatment2) -- node[above] {\tiny Treatment effect} (outcome2);
\node[format, above of=outcome2, dashed, left of=outcome2, node distance=2.2cm](unmeasured2){\scriptsize Unmeasured confounder $u$};
\draw[->] (unmeasured2) edge[] node {} (treatment2);
\draw[->] (unmeasured2) edge[] node{} (outcome2);
\node[format,left of=treatment2, node distance=2.6cm] (IV) {\scriptsize Candidate IVs $z$};
\draw[->] (IV) edge node[swap,above]{\small $\gamma$}(treatment2);
\draw[->] (IV) edge[bend right=25] node [swap, below=1mm]{\small $\kappa_z \neq \bm{0}$ }(outcome2);
\draw[-] (IV) edge[bend left=20] node [swap, below=2mm]{\small $\eta \neq \bm{0}$}(unmeasured2);
\end{tikzpicture}
\caption{Relaxation of assumptions $\kappa_z=0$ and \eqref{cond0-u} in Condition \ref{cond: valid IC}.} 
\label{fig: robust IV}
\end{figure}

In view of (\ref{cond1-u}), the conditional mean of the outcome can be written as
 \begin{equation}
 \label{eq2-2}
    \E[y_i|d_i,w_i,v_i]=\int q(d_i\beta+w_i^{\intercal}\kappa,u_i)f_u(u_i|w_i^{\intercal}\eta,v_i)du_i:=g^*\left(d_i \beta+w_i^{\intercal}\kappa, w_i^{\intercal}\eta,v_i\right)
 \end{equation}
for $g^*:\R^{p_{\eta}+2}\rightarrow \R$. 
In comparison to \eqref{Valid-IC}, expression (\ref{eq2-2}) allows the direct effects $\kappa_z\neq \bm{0}$ and additional indices $w_i^{\intercal}\eta$, which are induced by the dependence of $u_i$ and $w_i^{\intercal}\eta$ as in (\ref{cond1-u}).

We introduce an extra identifiability condition, which requires a majority of the candidate IVs to be valid. 
We use $\mathcal{S}$ to denote the set of relevant IVs and $\mathcal{V}$ to denote the set of valid IVs, i.e.,
$$\mathcal{S}=\{1\leq j\leq p_z: \gam_j\neq 0\} \quad \text{and}\quad \calV=\{j\in\mathcal{S}:\kappa_j=0,\eta_{j,.}=\bm{0}\}.$$ 
The set $\mathcal{S}$ contains all candidate IVs strongly associated with the exposure. The set $\mathcal{V}$ is a subset of $\mathcal{S}$. Notice that the IVs in set $\mathcal{V}$ have no direct effects on the outcome and are independent of the unmeasured confounder $u_i$ conditioning on $(z_{i,\calV^c}^{\intercal},x_i^{\intercal},v_i)$. Hence, we refer to $\mathcal{V}$ as the set of valid IVs and $\mathcal{V}^c$ as the set of invalid IVs.

When the candidate IVs are possibly invalid, the main challenge is that the set $\mathcal{V}$ is unknown a priori in data analysis. The following identifiability condition is needed to identify the causal effect without prior knowledge of the set of valid IVs.
 \begin{condition} [Majority rule]
 \label{cond: majority}
 More than half of the relevant IVs are valid:
$$\left|\calV\right|>\left|\mathcal{S}\cap\calV^c\right|.$$
 \end{condition}
The majority rule assumes that more than half of the relevant IVs are valid but does not require prior knowledge of the set $\mathcal{V}$. The majority rule has been proposed in linear outcome models with invalid IVs \citep{Bowden16,Kang16,TSTH,Wind19}. In this work, we generalize the majority rule into a semi-parametric format.  

To summarize, Conditions \ref{cond: dim red} and \ref{cond: majority} are new identifiability conditions to identify causal effects in the semi-parametric outcome model \eqref{eq: potential} with possibly invalid IVs. These two conditions weaken Condition \ref{cond: valid IC} and better accommodate for practical applications. Below we review two existing relaxations of Condition \ref{cond: valid IC} considered in the literature.

\begin{remark}[Relaxation in probit models]
\label{remark-carlson}
Along the line of probit outcome models, \citet{Rivers88} assumed 
$
u_i|v_i,w_i\sim N(\rho_1v_i,\rho_2^2)
$
 in a latent variable model for binary outcomes, where $\rho_1,\rho_2$ are unknown parameters. Here $v_i$ is a valid control function according to Condition \ref{cond: valid IC} and $f_u(\cdot)$ is known. \citet{carlson2021relaxing} also considered a probit outcome model but, using our notations, relaxed the model in \citet{Rivers88} to
\begin{align}
\label{carlson-cond}
u_i|w_i,v_i\sim N(\gam_0^{\intercal}g_1(v_i,w_i),\exp(2\delta_0^{\intercal}g_2(v_i,w_i))),
\end{align}
where $g_1(\cdot)$ and $g_2(\cdot)$ are known vector valued functions and $\gam_0,\delta_0$ are unknown parameters. That is, $u_i$ can depend on $w_i$ given $v_i$ through a known parametric form. Under (\ref{carlson-cond}) and some other technical conditions, they demonstrated the identifiability of model parameters and the ASF$(d,w)$ in (\ref{eq-asf}). In contrast, our work allows the conditional distribution $f(u_i|w_i,v_i)$ to be an unknown function as in (\ref{cond1-u}) which includes the normal distribution (\ref{carlson-cond}) as a special case.
\end{remark}
\begin{remark}[``CF-LI'' relaxation]
\label{remark-cfli}
In semi-parametric outcome models, the so-called ``CF-LI'' assumption was considered in \citet{Rothe09} and discussed in \citet{carlson2021relaxing}. Formally, ``CF-LI'' assumes, in the current context,
\[
 u_i|v_i,w_i\sim u_i|v_i, d_i\beta+x_i^{\intercal}\kappa_x.
\]
Under this assumption, $\E[y_i|d_i,w_i,v_i]=g_0(d_i\beta+x_i^{\intercal}\kappa_x,v_i)$ has the same format as the one under Condition \ref{cond: valid IC}. By setting $p_{\eta}=1$, $\eta_j=\beta\gam_j$ for $j=1,\dots,p_z$, and $\eta_j=\beta\gam_j+\kappa_j$ for $j=p_z+1,\dots,p$ in (\ref{cond1-u}), the ``CF-LI'' is recovered by Condition \ref{cond: dim red}. Hence, the ``CF-LI'' assumption is much more stringent than Condition \ref{cond: valid IC}. 
\end{remark}

 \subsection{Causal effects identification }
\label{sec: identification}

Under Condition \ref{cond: dim red}, we generalize the original ASF$(d,w)$ in (\ref{eq-asf}) and define
 \begin{align}
\label{eq-phi}
\phi^*(d,w)=\E[\E[q(d\beta+w^{\intercal}\kappa,u_i)|w_i=w,v_i]]=\E[g^*(d\beta+w^{\intercal}\kappa,w^{\intercal}\eta,v_i)],
\end{align}
where $g^*$ is defined in \eqref{eq2-2}.
The definition of $\phi^*(d,w)$ directly generalizes the last expression in (\ref{eq-asf2}). Notice that $\phi^*(d,w)=\textup{ASF}(d,w)$ under Condition \ref{cond: valid IC} but they can be different under Condition \ref{cond: dim red}.
The quantity $\phi^*(d,w)$ is of interest for two reasons. First,  Condition \ref{cond: dim red} allows part of $u_i$ to be explained by the measured variables $w_i$ through the form $w_i^{\intercal}\eta$. In (\ref{eq-phi}), we fix the observed variables at the given values $(d,w^{\intercal})$ and only average out the truly unobserved parts. Second, the original ASF$(d,w)$ has identifiability issues when Condition \ref{cond: valid IC} is violated; see Corollary 4.1 in \citet{carlson2021relaxing}.
In contrast, $\phi^*(d,w)$ can be identified via a partial mean, as described in the following.

Importantly, if $w_i$ and $v_i$ are independent, then the conditional average causal effect(CATE) relates to $\phi^*(d,w)$  through the following expression,
\begin{align}
\label{eq: relation}
 \textup{CATE}(d,d'|w):=\E[y_i^{(d)}|w_i=w]-\E[y_i^{(d')}|w_i=w]=\phi^*(d,w)-\phi^*(d',w).
\end{align}

We now describe how to identify $\phi^*(d,w)$ defined in (\ref{eq-phi}) and will present the data-dependent algorithm in Section \ref{sec3}. We rewrite  the conditional mean function (\ref{eq2-2}) as
\begin{equation}
\label{eq-cond}
\E[y_i|d_i,w_i,v_i]=g^*((d_i,w_i^{\intercal})B^*,v_i) \quad \text{with}\quad B^*=\begin{pmatrix}
\beta & \bm{0} \\
\kappa & \eta 
\end{pmatrix}\in\R^{(p+1)\times (p_{\eta}+1)}.
\end{equation}
Due to the collinearity among $d_i, w_i$, and $v_i$, we cannot directly identify $B^*$ in (\ref{eq-cond}). Instead, we apply $\E[y_i|w_i,v_i]=\E[y_i|d_i,w_i,v_i]$ and derive the following reduced-form representation by combining \eqref{eq-d} and \eqref{eq-cond}, which is
\begin{equation}
\label{eq-Theta}
    \E[y_i|w_i,v_i]=\E[y_i|w_i^{\intercal}\Theta^*, v_i]\quad \text{with}\quad\Theta^*=(\gam, {\rm I}_p)B^*\in\R^{p\times (p_{\eta}+1)}
\end{equation}
and ${\rm I}_p$ being the $p\times p$ identity matrix.

In the rest of this section, we assume $\Theta^*,\gam$ and $\mathcal{S},$ can be accurately estimated and describe how to identify the model parameters $B^*$ and the functional $\phi^*(d,w)$ with $\Theta^*$.   In Section \ref{sec3-1}, we will construct an estimator of $\Theta^*$, which is closely related to the estimation of the central subspace or central mean space in the semi-parametric literature \citep[e.g.]{Cook02,cook98}. We provide data-dependent estimators of $\gam$ and $\mathcal{S}$ in the following equations \eqref{est1} and \eqref{eq-Shat}, respectively. We apply the majority rule (Condition \ref{cond: majority}) and identify the matrix $B^*$ by the expression $\Theta^*=(\gam, {\rm I}_p)B^*.$ 
With $\Theta^*,\gam,$ and $\mathcal{S},$ we define $$b^*_{m}= {\rm Median}(\{\Theta^*_{j,m}/\gam_j\}_{j\in\mathcal{S}}) \quad \text{for}\quad m=1,\dots,p_{\eta}+1.$$  We further identify $B^*$ as
\begin{equation}
\label{eq-B-ident}
 B^*=\begin{pmatrix}
 b^*_{1} &\dots&b^*_{p_{\eta}+1} \\
 \Theta^*_{.,1}-b_{1}\gam &\dots&\Theta^*_{.,p_{\eta}+1}-b^*_{p_{\eta}+1}\gam
 \end{pmatrix},
 \end{equation}
where $\Theta^*_{.,j}$ denotes the $j$-th column of $\Theta^*$. The rationale for \eqref{eq-B-ident} is the same as the application of majority rule in linear outcomes models: each candidate IV can produce an estimate of the causal effect $\beta$ based on the ratio of the reduced-form parameter and the IV strength $\gam$; the median of these ratios will be $\beta$ if more than half of the relevant IVs are assumed to be valid. The definition of $B^*$ in \eqref{eq-B-ident} generalizes the median idea \citep{Bowden16,Kang16,Wind19} to semi-parametric outcome models.

To identify the functional $\phi^*(d,w)$ in (\ref{eq-phi}), we  derive its integrand based on (\ref{eq: connection}) and (\ref{eq-cond}), which gives
 \begin{equation}
\label{eq2-7}
\E\left[y_i^{(d)}|w_i=w, v_i=v\right]=g^*\left((d,w^{\intercal})B^*,v\right).
\end{equation}
Then the quantity $\phi^*(d,w)$  can be identified by taking an integration of $g^*((d,w^{\intercal})B^*,v_i)$ with respect to the density of $v_i$. The CATE can be identified via its linear relationship with $\phi^*$ function as in \eqref{eq: relation}.

\section{Methodology: SpotIV}
\label{sec3}
In this section, we propose inference methods for the causal functional $\phi^*(d,w)$ in \eqref{eq-phi} and $\textup{CATE}(d,d'|w)$ in \eqref{eq: relation}. We detail our proposed procedure binary outcomes in Section \ref{sec3-2}, \ref{sec3-1}, and \ref{sec3-3}. Verification of the majority rule is considered in Section \ref{sec3-5}.
The proposed method has three steps. 

\subsection{Step 1: Estimation of reduced-form parameters} 
\label{sec3-2}
We first fit the first-stage model (\ref{eq-d}) based on least squares,
\begin{equation}
\label{est1}
\hat{\gamma}=(W^{\intercal}W)^{-1}W^{\intercal}d\quad\text{and}\quad \hat{v}=d-W\widehat{\gam}.
\end{equation}
In the following, we detail a specific estimator of the reduced form parameter $\Theta^*$ in \eqref{eq-Theta}. The procedure is derived from the sliced-inverse regression (SIR) approach proposed in \citet{Li91}. To facilitate the discussion, we restrict our attention to the binary outcome model and will discuss the extension to general outcome in the following Remark \ref{rem: general outcome}.

Let $\Sig=\E[w_iw_i^{\intercal}]\in\R^{p\times p}$ denote the covariant matrix of observed covariates and $\widehat{\Sig}=\sum_{i=1}^nw_iw_i^{\intercal}/n$ denote the empirical estimate of $\Sig$. 
Define 
\begin{equation}
\Omega=\text{Cov}(\alpha(y_i)) \in \R^{p\times p} \quad \text{with} \quad \alpha(y_i)= \E[\Sig^{-1/2}w_i|y_i]\in \R^{p}
\label{eq: cov def}
\end{equation}
where $\alpha(y_i)$ denotes the inverse regression function and $\Omega$ denotes its covariance matrix with the expression 
$\Omega={\P}(y_i=1){\P}(y_i=0)\{{\alpha}(1)-{\alpha}(0)\}\{{\alpha}(1)-{\alpha}(0)\}^{\intercal}.$

For $k=0,1$, we estimate $\alpha(k)$ by
$$
  \hat{\alpha}(k)=\frac{1}{\sum_{i=1}^n \mathbbm{1}(y_i=k)}\sum_{i=1}^n\mathbbm{1}(y_i=k)\widehat{\Sig}^{-1/2}w_i
$$ 
and estimate $\Omega$ by
$
 \widehat{\Omega}=\widehat{\P}(y_i=1)\widehat{\P}(y_i=0)\{\hat{\alpha}(1)-\hat{\alpha}(0)\}\{\hat{\alpha}(1)-\hat{\alpha}(0)\}^{\intercal}
$
with $\widehat{\P}(y_i=1)=\sum_{i=1}^n\mathbbm{1}(y_i=1)/n$ and $\widehat{\P}(y_i=0)=1-\widehat{\P}(y_i=1)$.
Let $\hat{\lam}_1\geq \dots\geq \hat{\lam}_{p}$ denote the eigenvalues of $\widehat{\Omega}$ and $(\widehat{\phi}_1,\dots,\widehat{\phi}_p)\in\R^{p\times p}$ denote the matrix of the eigenvectors of $\widehat{\Omega}$ corresponding to $\hat{\lam}_1,\dots,\hat{\lam}_{p}$. We estimate $\Theta^*$ using the eigenvectors corresponding to the nonzero eigenvalues of $\widehat{\Omega}$. Let $M$ denote the rank of $\Omega$ and $\widehat{M}$ be an estimate of $M$. Define the estimator of $\Theta^*$ as 
\begin{align}
\label{eq-hTheta}
 \widehat{\Theta}=(\widehat{\phi}_{1},\dots\widehat{\phi}_{\widehat{M}}).
\end{align}
For estimating $M$, a BIC-type procedure in \citet{Zhu06} can be applied. Specifically, we consider
{\small
\begin{equation}
\label{eq-hM}
  \widehat{M}=\argmax_{1\leq m\leq p} C(m)~~ \text{with}\; C(m)=\frac{n}{2}\sum_{i=m+1}^{p}\{\log (\hat{\lam}_i+1)-\hat{\lam}_i\}\mathbbm{1}(\hat{\lam}_i>0)-\frac{n^{c_0} \cdot m(2p-m+1)}{2},
\end{equation}
}
where $n^{c_0}$, with $0<c_0<1$, is a penalty constant and $m(2p-m+1)/2$ is the degree of freedom. The consistency of $\widehat{M}$ follows from Theorem 2 in \citet{Zhu06}. 

Since the SIR approach estimates a basis of the linear space of $\Theta^*$, the probabilistic limit of $\widehat{\Theta}$ is indeed a linear transformation of $\Theta^*$. As we will formally prove in the next section (Lemma \ref{lem2-semi}), the proposed method is invariant to linear transformations. Consistency and asymptotic normality of the proposed estimator can be established under any fixed linear transformation.

\subsection{Step 2: estimation of $B^*$ based on SIR}
\label{sec3-1}

To apply the majority rule, we first select the set of relevant IVs by
\begin{equation}
\widehat{\mathcal{S}}=\left\{1\leq j\leq p_z: {\left|\widehat{\gam}_j\right|} \geq \hat{\sig}_v\sqrt{2\{\widehat{\Sig}^{-1}\}_{j,j} \log n/n } \right\},
\label{eq-Shat}
\end{equation}
where $\hat{\sig}^2_v=\sum_{i=1}^n\hat{v}_i^2/n$ with $\hat{v}_i$ defined in (\ref{est1}). 
 The term $\log n$ is the adjustment for the multiplicity of the selection procedure in \eqref{eq-Shat}. Under mild conditions, $\widehat{\mathcal{S}}$ can be shown to be a consistent estimator of $\mathcal{S}$. Within $\widehat{\mathcal{S}}$, we apply the median rule to estimate $B^*$ according to (\ref{eq-B-ident}). Specifically, for $m=1,\dots,\widehat{M}$ we define $\hat{b}_m={\rm Median}\left(\left\{\widehat{\Theta}_{j,m}/\widehat{\gam}_j\right\}_{j\in\widehat{\mathcal{S}}}\right)$ for $m=1,\dots,\widehat{M}$ and 
  \begin{align}
 & \widehat{B}=\begin{pmatrix}\hat{b}_1 &\dots&\hat{b}_{\widehat{M}} \\
 \widehat{\Theta}_{.,1}- \hat{b}_1\widehat{\gam} &\dots & \widehat{\Theta}_{.,\widehat{M}}- \hat{b}_{\widehat{M}} \widehat{\gam} 
 \end{pmatrix}.
 \label{eq-Bhat}
  \end{align}

\subsection{Step 3: inference for the causal estimands}
\label{sec3-3}
We introduce inference procedures for $\phi^*(d,w)$ defined in \eqref{eq-phi}. In view of (\ref{eq2-7}), after identifying the parameter matrix $B^*$, we consider estimating the unknown function $g^*(\cdot)$. With $\widehat{B}$ defined in \eqref{eq-Bhat}, we estimate $g^*$ by a kernel estimator $\widehat{g}$.
Denote the estimated indices as $\widehat{s}_i=((d,w^{\intercal})\widehat{B},\widehat{v}_i)^{\intercal}\in\R^{\widehat{M}+1}$ and $\widehat{t}_i=((d_i,w_i^{\intercal})\widehat{B},\widehat{v}_i)^{\intercal}\in\R^{\widehat{M}+1}$, for $1\leq i\leq n$. Define the kernel $K_{H}(a, b)$ for $a, b \in \R^{\widehat{M}+1}$ as 
$K_{H}(a, b)=\prod_{l=1}^{\widehat{M}+1} \tfrac{1}{h_{l}}k\left(\frac{a_{l}-b_{l}}{h_l}\right)
$
where $h_l$ is the bandwidth for the $l$-th argument and $k(x)={\bf 1}\left(|x|\leq 1/2\right).$ For the sake of illustration, we take $K_{H}$ in the form of product kernel and $k(x)$ as the box kernel and set $h_l=h$ for $1\leq l\leq \widehat{M}+1$. Our proposed method can be extended to allow for a more general form of kernel function.

We construct the following kernel estimators of $\{g(s_i)\}_{1\leq i\leq n}$, 
$$\widehat{g}(\widehat{s}_i)=\frac{\frac{1}{n}\sum_{j=1}^{n}y_j K_{H}(\widehat{s}_i, \widehat{t}_j)}{\frac{1}{n}\sum_{j=1}^{n} K_{H}(\widehat{s}_i, \widehat{t}_j)} \quad \text{for}\quad 1\leq i\leq n.$$
We apply the partial mean methods and further estimate $\phi^*(d,w)=\int g^*(s_i) f_v(v_i)d v_i$ in (\ref{eq-phi}) by 
\[
\widehat{\phi}(d,w)=\frac{1}{n} \sum_{i=1}^{n} \widehat{g}(\widehat{s}_i).
\] 

We estimate $\phi^*(d',w)$ analogously and then estimate ${\rm CATE}(d,d'|w)$ as
\begin{equation} 
\label{eq-cate-est}
\widehat{\rm CATE}(d,d'|w)=\widehat{\phi}(d,w)-\widehat{\phi}(d',w).
\end{equation}
In Section \ref{sec-inf}, we establish the asymptotic normality of $\widehat{\rm CATE}(d,d'|w)$. We approximate its variance by bootstrap and construct the confidence interval for ${\rm CATE}(d,d'|w)$ as 
\begin{equation}
\label{eq-ci-ate}
\left( \widehat{\rm CATE}(d,d'|w)- z_{1-{\alpha}/{2}}\widehat{\sig}^*,\quad \widehat{\rm CATE}(d,d'|w)+ z_{1-{\alpha}/{2}}\widehat{\sig}^*\right),
\end{equation}
where $z_{1-\alpha/2}$ is the $1-\alpha/2$ quantile of standard normal and $\widehat{\sig}^*$ is the standard deviation estimated by $N$ bootstrap samples. 

\begin{remark}[Extension to continuous nonlinear outcome models]
\rm The SpotIV procedure for binary outcomes detailed above can be extended to deal with continuous nonlinear outcome models. The main change is to use a different estimator of the covariance matrix $\Omega=\text{Cov}(\alpha(y_i))$. Specifically, $\Omega$ can be estimated based on SIR \citep{Li91} or kernel-based method \citep{ZF96}. With such an estimate of $\Omega$, we can apply the same procedure as in Sections \ref{sec3-1} to \ref{sec3-3} and make inference for {\rm CATE}. We examine the numerical performance of our proposal for continuous nonlinear outcome models in the online supplementary materials \citep{supp2}. 
\label{rem: general outcome}
\end{remark}
\subsection{Testing the majority rule}
\label{sec3-5}
The majority rule (Condition \ref{cond: majority}) allows less than half of the relevant IVs to be invalid. It is crucial to know whether the majority rule is plausible in applications. Although the new identifiability conditions cannot be thoroughly tested in a data-dependent way, we describe in the following how to test the majority rule partially. The method is based on a ``voting'' idea derived from \citet{TSTH}.

We illustrate the idea using the true parameters. 
Define $b^{(j)}=\Theta^*_{j,1}/\gam_j$ for $j\in \mathcal{S}$. 
If $j\in\calV$ and $k\in \calV$, we have $\Theta^*_{k,1}-b^{(j)}\gam_k=0$. That is, if both the $j$-th and $k$-th IVs are valid, they will vote for each other to be valid. Let $C_k=|\{j\in\mathcal{S}: \Theta^*_{k,1}-b^{(j)}\gam_k=0\}|$ denote the number of votes received by the $k$-th IV.
If the majority rule (Condition \ref{cond: majority}) holds, then all valid IVs will receive more than $|\mathcal{S}|/2$ votes, that is, 
$$|\{k\in\mathcal{S}: C_k>|\mathcal{S}|/2\}|\leq |\mathcal{S}|/2.$$ However, if $|\{k\in\mathcal{S}: C_k>|\mathcal{S}|/2\}|\leq |\mathcal{S}|/2$, then the majority rule fails. 
We comment that only the first column of $\Theta^*$ is used in the definition of $C_k$ since $\Theta^*_{.,1}$ is the most significant direction.  
To account for the uncertainty in the data, we estimate $C_k$ by \[
 \widehat{C}_k=|\{j\in \widehat{\mathcal{S}}: |\widehat{\Theta}_{k,1}-\widehat{b}^{(j)}\widehat{\gam}_k|\leq \eps_n^{(j,k)}\}|,
\]
where $\eps_n^{(j,k)}$ goes to zero as $n\rightarrow \infty$. We set $\eps_n^{(j,k)}$ as in (7) of \citet{TSTH}, which is based on the asymptotic covariance of $\widehat{\Theta}_{.,1}$ and $\widehat{\gam}$. For $\eps_n^{(j,k)}$, we replace  (7) of \citet{TSTH} by its binary outcome counterpart. Specifically,
\begin{align}
\label{eq-epsn}
\eps_n^{(j,k)}=2.01\frac{\|W(\widehat{U}_{.k}-\frac{\hat{\gam}_k}{\hat{\gam}_j}\widehat{U}_{.j})\|_2}{\sqrt{n}}\sqrt{\frac{\log \max\{p_z,n\}}{n}},
\end{align}
where $\widehat{U}=\{\frac{1}{n}\sum_{i=1}^nw_iw_i^{\intercal}\hat{p}_i(1-\hat{p}_i)\}^{-1}$ and $\hat{p}_i=\hat{g}(w_i^{\intercal}\widehat{\Theta}+\hat{v}_i)$ is the estimate of $\E[y_i|w_i,v_i]$. Here, $\hat{g}$ can be obtained based on the kernel estimates as above. In practice, one can also approximate $\E[y_i|w_i,v_i]$ by fitting logistic or probit models in this thresholding step.
We demonstrate the performance of this partial check of the majority rule in Section \ref{ap-simu-3}.
\section{Theoretical justifications}
\label{sec4}
We provide theoretical justifications of our proposed method for binary outcome models. In Section \ref{sec4-1}, we present the estimation accuracy of the model parameter matrix $\widehat{B}$. In Section \ref{sec-inf}, we establish the asymptotic normality of the proposed SpotIV estimator under proper conditions. 
\subsection{Estimation accuracy of the model parameter matrix} 
\label{sec4-1}
We now introduce the required regularity conditions.
\begin{condition}
\label{cond-semi}
\textup{(Data distribution)} The observed data $(y_i, d_i, w_i^{\intercal})$, $i=1,\dots,n$, are \textit{i.i.d.} generated with $y_i\in\{0,1\}$ and $\E[w_iw_i^{\intercal}]$ being positive definite. The dimension of $\eta$ defined in Condition \ref{cond: dim red} satisfies $p_{\eta}=1.$ The covariates $\{w_{i,j}\}_{1\leq j\leq p}$ and the exposure $d_i$ are sub-Gaussian random variables.  
\end{condition}

In Condition \ref{cond-semi}, we assume binary outcomes, sub-Gaussian exposure, and sub-Gaussian covariates. To simplify the technical proofs for the nonparametric estimation, we focus on $p_{\eta}=1$ in (\ref{cond1-u}), that is, the invalid effect through unmeasured confounder has rank one.

\begin{condition}[Regularity conditions for SIR]
\label{cond-semi1}
The conditional mean $\E[w_i|w_i^{\intercal}\Theta^*,v_i]$ is linear in $w_i^{\intercal}\Theta^*$ and $v_i$. For the covariance matrix $\Omega$ defined in \eqref{eq: cov def}, its rank $M$ satisfies $M=p_{\eta}+1$, and the nonzero eigenvalues are simple.
\end{condition}

Condition \ref{cond-semi1} assumes a linearity assumption, which is standard for SIR methods \citep{Li91, CL99, CCL02}. 
A sufficient condition for the linearity assumption is that $w_i$ is normal and is independent of $v_i$. \citet{hall1993almost} shows that the linearity condition always offers an excellent approximation to the reality when $p$ diverges to infinity while the dimension of central space remains fixed. 
Invoking that $\Theta^*$ has $p_{\eta}+1$ columns, we assume $M=p_{\eta}+1$ to exclude some degenerated cases. 
The simple nonzero eigenvalues guarantee the uniqueness of its eigenvector matrix. Similar assumptions have been imposed in \citet{ZF96}.

The next lemma establishes the convergence rate of $\widehat{B}$. The probabilistic limit of $\widehat{B}$ can be expressed as follows. Let $\Theta\in\R^{p\times M}$ denote the eigenvectors of $\Omega=\text{Cov}(\alpha(y_i))$ corresponding to all the nonzero eigenvalues of $\Omega$. 
Define \begin{equation}
\label{eq-B-ap}
 B=\begin{pmatrix}
 b_{1}  &\dots &b_{M} \\
 \Theta_{.,1}-b_{1}\gam &\dots&\Theta_{.,M}-b_{M}\gam 
 \end{pmatrix}
 \end{equation}
with $b_{m}= {\rm Median}(\{\Theta_{j,m}/\gam_j\}_{j\in\mathcal{S}})$, $ m=1,\dots,M$. 
 \begin{lemma}
\label{lem2-semi}
Assume Conditions \ref{cond: dim red}, \ref{cond: majority}, \ref{cond-semi}, and \ref{cond-semi1} hold. Then for some positive constants $c_1, c_2>0$, it holds that
\begin{align}
\P\left(\|\widehat{B}-B\|_2\geq c_1\sqrt{{t}/{n}}\right)\leq \exp(-c_2t)+\P(E_1^c)
\label{eq: B-accuracy}
\end{align}
where $\P(E_1^c)\rightarrow 0$ as $n\rightarrow \infty$,
Moreover, $\E[y_i|d_i,w_i^{\intercal},v_i]=\E[y_i|(d_i,w_i^{\intercal})B,v_i]$. 
\end{lemma}
This lemma implies that $\widehat{B}$ converges to $B$ at rate $n^{-1/2}$ and $(d_i,w_i^{\intercal})B,v_i$ also provides a sufficient summary for the conditional mean of the outcome. We see that $B$ in (\ref{eq-B-ap}) has the same format as $B^*$ except that $\Theta^*$ is replaced with $\Theta$. In the proof, we show that $B$ is a linear transformation of $B^*$, where the transformation corresponds to the rotation of $\Theta^*$ to $\Theta$. 
The high probabilistic event $E_1$ is defined as $\{\widehat{\mathcal{S}}=\mathcal{S},\widehat{M}=M\}$. 
As a remark, the result in Lemma \ref{lem2-semi} still holds if the estimator $\widehat{\Theta}$ is replaced with other $\sqrt{n}$-consistent estimators of $\Theta$. 

\subsection{Asymptotic normality}
\label{sec-inf}
In the following, we establish the asymptotic normality of our proposed SpotIV estimator. Let $t_i=((d_i,w_i^{\intercal}){B},v_i)^{\intercal}\in \R^3$ and $s_i=((d,w^{\intercal})B,v_i)^{\intercal}\in \R^3$. The dimension three comes from the assumption that $M=p_{\eta}+1=2$. Define \begin{equation}
\mathcal{N}_{h}(s)=\left\{t\in \R^{3}: \|t-s\|_{\infty}\leq h \right\}, 
\label{eq: neighborhood}
\end{equation}
where $\|\cdot\|_{\infty}$ denotes the vector maximum norm.

\begin{condition}[Smoothness conditions]
\label{cond-semi2}
\begin{enumerate}
\item[(a)] The density function $f_t$ of $t_i=((d_i,w_i^{\intercal}){B},v_i)^{\intercal}$ has a convex support $\mathcal{T}\subset \R^3$ and satisfies {$c_0 \leq f_t(s_i)\leq C_0$ for all $1\leq i\leq n$}, $\int_{t\in \mathcal{T}^{\rm int}} f_t(t)dt=1$ and {$\max_{1\leq i\leq n}\sup_{t\in \mathcal{N}_{h}(s_i)\cap \mathcal{T}}\|\triangledown f_t(t)\|_{\infty}\leq C$}, where $\mathcal{T}^{\rm int}$ is the interior of $\mathcal{T}$, $\mathcal{N}_{h}(s)$ is defined in \eqref{eq: neighborhood}, $\triangledown f_t$ is the gradient of $f_t$ and $C_0>c_0>0$ and $C>0$ are positive constants. The density $f_v$ of $v_i$ is bounded from above and has a convex support $\mathcal{T}_{v}.$  
\item[(b)] The function $g$ defined in \eqref{eq-cond} is twicely differentiable. For any $1\leq i\leq n$, $g(s_i)$ is bounded away from zero and one. The function $g$ satisfies $\max_{1\leq i\leq n}\sup_{t\in \mathcal{N}_{h}(s_i)\cap \mathcal{T}}\\
\|\triangledown g(t)\|_{2}\leq C$ and $\max_{1\leq i\leq n}\sup_{t\in \mathcal{N}_{h}(s_i)\cap \mathcal{T}}\lambda_{\max}(\triangle g(t))\leq C$, where $\mathcal{N}_{h}(s)$ is defined in \eqref{eq: neighborhood}, $\|\triangledown g(t)\|_2$ and $\lambda_{\max}(\triangle g(t))$ respectively denote the $\ell_2$ norm of the gradient vector and the largest eigenvalue of the hessian matrix of $g$ evaluated at $t$ and $C>0$ is a positive constant.
\item[(c)] For any $v\in \mathcal{T}_{v}$, then the evaluation point $(d,w^{\intercal})^{\intercal}$ satisfies $((d, w^{\intercal})B+\Delta^{\intercal},v)^{\intercal}\in \mathcal{T}$ for any $\Delta\in \R^2$ and $\|\Delta\|_{\infty}\leq h$. 
\end{enumerate}
\end{condition}

Condition \ref{cond-semi2}(a) and \ref{cond-semi2}(b) are mainly imposed for the regularities of the density function $f_t$, $f_v$, and the conditional mean function $g$ at $s_i=((d,w^{\intercal})B,v_i)^{\intercal}$ or its neighborhood $\mathcal{N}_{h}(s_i)$. Here the randomness of $s_i$ only depends on $v_i$ for the pre-specified evaluation point $(d,w^{\intercal})^{\intercal}$. Condition \ref{cond-semi2}(c) essentially assumes that the evaluation point $(d,w^{\intercal})$ is not at the tail of the joint distribution of $(d_i,w_i^{\intercal}).$ In the online supplement \citep{supp2},  we verify Condition \ref{cond-semi2} in some generic examples. Specifically, we will verify Condition \ref{cond-semi2} (a) under the regularity conditions on the density function of $t_i^*$. Condition \ref{cond-semi2} (b) is guaranteed by the regularity conditions on the potential outcome model $q(\cdot)$ defined in \eqref{eq: potential} and the  conditional density $f_u(u_i|w_i^{\intercal}\eta,v_i)$. If $q(\cdot)$ is continuous, it suffices to require that $q(\cdot)$ has bounded second derivatives and the density $f_u(u_i|w_i^{\intercal}\eta,v_i)$ belongs to a location-scale family with smooth mean and variance functions. If $q(\cdot)$ is an indicator function, then $g$ becomes the conditional density of $u_i$ given $w_i^{\intercal}\eta$ and $v_i$ and it suffices to require this conditional density function to satisfy Condition \ref{cond-semi2} (b). Examples of $q$ functions satisfying Condition \ref{cond-semi2} (b) include logistic or probit models with uniformly bounded $v_i$.

The following theorem establishes the asymptotic normality of the proposed estimator of $\phi^*(d,w)$. 
\begin{theorem} 
Suppose that Condition \ref{cond-semi2} holds, and the bandwidth satisfies $h=n^{-\mu}$ for $0<\mu<1/4$. For any estimator $\widehat{B}$ satisfying \eqref{eq: B-accuracy}, there exists positive constants $c>0$ and $C>0$ such that, with probability larger than $1-n^{-c}-\P(E_1^c)$,
$
\left|\widehat{\phi}(d,w)-\phi^*(d,w)\right|\leq C\left(\frac{1}{\sqrt{nh^2}}+ h^2\right),
$
where $\P(E_1^c)\rightarrow 0$ as $n\rightarrow \infty$.
Taking $h=n^{-\mu}$ for any $\mu\in(0,1/6)$, we have 
$$
\frac{n}{\sqrt{\rm V}}\left(\widehat{\phi}(d,w)-\phi^*(d,w)\right)\cid N(0,1) \quad \text{with}\quad {\rm V}=\sqrt{\sum_{j=1}^{n} a_j^2 g({t}_j)(1-g({t}_j))}
$$
where $a_j=\frac{1}{n} \sum_{i=1}^{n} \frac{ K_{H}({s}_i, {t}_j)}{\frac{1}{n}\sum_{j=1}^{n} K_{H}({s}_i, {t}_j)} \; \text{for}\; 1\leq j\leq n$ and $\cid$ denotes the convergence in distribution. 
There exist some positive constants $C_0\geq c_0>0$ and $c>0$ such that the asymptotic standard error satisfies
$
\P\left(c_0/\sqrt{nh^2}\leq \sqrt{{\rm V}}/{n}\leq C_0/\sqrt{nh^2}\right)\geq 1-n^{-c}.
$
\label{thm: asf}
\end{theorem}
A few remarks are in order for this main theorem. Firstly, the rate of convergence for $\widehat{\phi}(d,w)$ is the same as the optimal rate of estimating a twice-differentiable function in two dimensions \citep{tsybakov2008introduction}. 
Though the unknown target function $\phi^*(d,w)$ can be viewed as a two-dimension function on linear combinations of $d$ and $w$, it cannot be directly estimated using the classical nonparametric methods. In contrast, we have first to estimate the unknown function $g$ in three dimensions and then further estimate the target $\phi(d,w)$. After a careful analysis, we establish that, even though $\phi(d,w)$ involves estimating the three-dimension function $g$, the final convergence rate can be reduced to the same rate as estimating two-dimensional twice-differentiable smooth functions. This type of result has been established in \citet{Newey94} and \citet{Linton95} under the name ``partial mean". However, our proof is distinguished from the standard partial mean problem in the sense that we do not have access to direct observations of $s_i$ and $t_i$ but only have their estimators $\widehat{s}_i$ and $\widehat{t}_i$ for $1\leq i\leq n.$ 

Secondly, beyond Condition \ref{cond-semi2}, the above theorem requires a suitable bandwidth condition $h=n^{-\mu}$ with $0<\mu<1/6$ for establishing the asymptotic normality, which is standard in nonparametric regression in two dimensions \citep{wasserman2006all}. This bandwidth condition requires the variance component to dominate its bias, that is, $(nh^2)^{-1/2}\gg h^2.$ Thirdly, the asymptotic normality holds for a large class of initial estimators $\widehat{B}$ as long as they satisfy \eqref{eq: B-accuracy}. By Lemma \ref{lem2-semi}, our proposed estimator $\widehat{B}$ belongs to this class of initial estimators with a high probability.

Similar to the definition of $s_i$, we define ${r}_i=((d', w^{\intercal})B, v_i)$ as the corresponding multiple indices by fixing $(d_i, w_i^{\intercal})$ at the given level $(d',w^{\intercal})$. The following corollary establishes the asymptotic normality of the proposed estimator $\widehat{\rm CATE}(d,d'|w)$ defined in \eqref{eq-cate-est}.
\begin{corollary}
Suppose that Condition \ref{cond-semi2} holds for both $\{s_i\}_{1\leq i\leq n}$ and replacing $\{s_i\}_{1\leq i\leq n}$ and $d$ by $\{{r}_i\}_{1\leq i\leq n}$ and $d'$, respectively. Suppose that, $v_i$ is independent of $w_i$, the bandwidth satisfies $h=n^{-\mu}$ for $\mu\in(0,1/6)$, and $\left|d-d'\right|\cdot \max\{|B_{11}|,|B_{21}|\}\geq h.$ For any estimator $\widehat{B}$ satisfying \eqref{eq: B-accuracy}, then  
$$
\frac{n}{\sqrt{{\rm V}_{\rm CATE}}}\left(\widehat{\rm CATE}(d,d'|w)-{\rm CATE}(d,d'|w)\right)\cid N(0,1) ,
$$
where ${\rm V}_{\rm CATE}=\sqrt{\sum_{j=1}^{n} c_j^2 g({t}_j)(1-g({t}_j))}
$ for $c_j=\frac{1}{n} \sum_{i=1}^{n} \left(\frac{ K_{H}({s}_i, {t}_j)}{\frac{1}{n}\sum_{j=1}^{n} K_{H}({r}_i, {t}_j)}- \frac{ K_{H}({r}_i, {t}_j)}{\frac{1}{n}\sum_{j=1}^{n} K_{H}({r}_i, {t}_j)}\right)$, $1\leq j\leq n.$ 
There exist some positive constants $C_0\geq c_0>0$ and $c>0$ such that 
$
\P\left(c_0/\sqrt{nh^2}\leq \sqrt{{\rm V}_{\rm CATE}}/{n}\leq C_0/\sqrt{nh^2}\right)\geq 1-n^{-c}.$
\label{cor: cate}
\end{corollary}

Corollary \ref{cor: cate} is closely related to Theorem \ref{thm: asf}. The asymptotic normality of $\phi^*(d',w)$ can be established with a similar argument to Theorem \ref{thm: asf} with replacing $s_i$ by $r_i$. When $v_i$ is independent of the measured covariates $w_i$, we apply \eqref{eq: relation} to compute ${\rm CATE}(d,d'|w)$ by taking the difference of $\widehat{\phi}(d,w)$ and $\widehat{\phi}(d',w)$. An extra step is to show that the asymptotic normal component of $\widehat{\phi}(d,w)-\widehat{\phi}(d',w)$ dominates its bias component. To ensure this, an extra assumption on the difference between $d$ and $d'$, $\left|d-d'\right|\cdot \max\{|B_{11}|,|B_{21}|\}\geq h$, is needed to guarantee the lower bound for $\sqrt{{\rm V}_{\rm CATE}}/{n}$. 

\section{Numerical studies}
\label{sec-simu}
In this section, we assess the empirical performance of the proposed method for both binary and continuous outcome models. 
For implementation, we estimate $\widehat{\Theta}$ using the SIR method in the R package \texttt{np} \citep{np}. The bandwidth is set by rule of thumb $h_k=0.9\min\{\hat{\sig}_k,\textup{IQR}_k/1.34\}n^{-1/(5+\widehat{M})}$ where $\hat{\sig}_k$ is the standard deviation of the $k$-th index and $\textup{IQR}_k$ is the interquartile range of the $k$-th index  for $k=1,\dots,\widehat{M}+1$.
To construct confidence intervals for CATE, we use the standard deviation of $N=50$ bootstrap realizations to estimate its standard error.

We consider two simulation scenarios with no measured covariates $x_i$, i.e., $w_i=z_i$, and set $p=p_z=7$. Setting (i) generates binary outcomes and setting(ii) generates continuous nonlinear outcomes The estimand $\phi^*$ and the CATE functions are nonlinear in these scenarios. We plot their corresponding $\phi^*(d,w)$ (as a function of $d$) in Figure \ref{fig1}. The R code and further simulation results are available at \url{https://github.com/saili0103/SpotIV}.

\begin{figure}[H]
\centering
\includegraphics[width=0.35\textwidth, height=4cm]{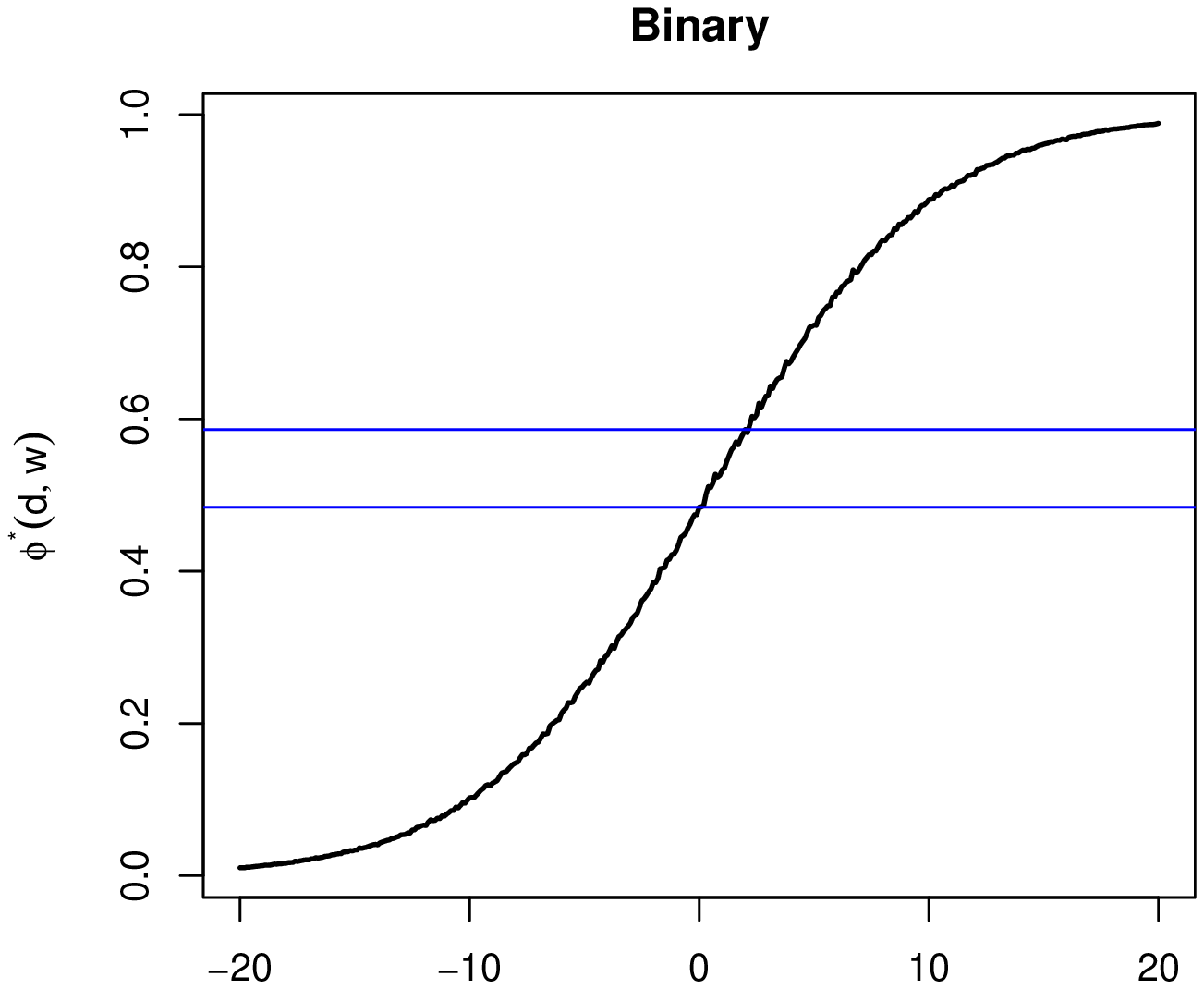}
\includegraphics[width=0.35\textwidth, height=4cm]{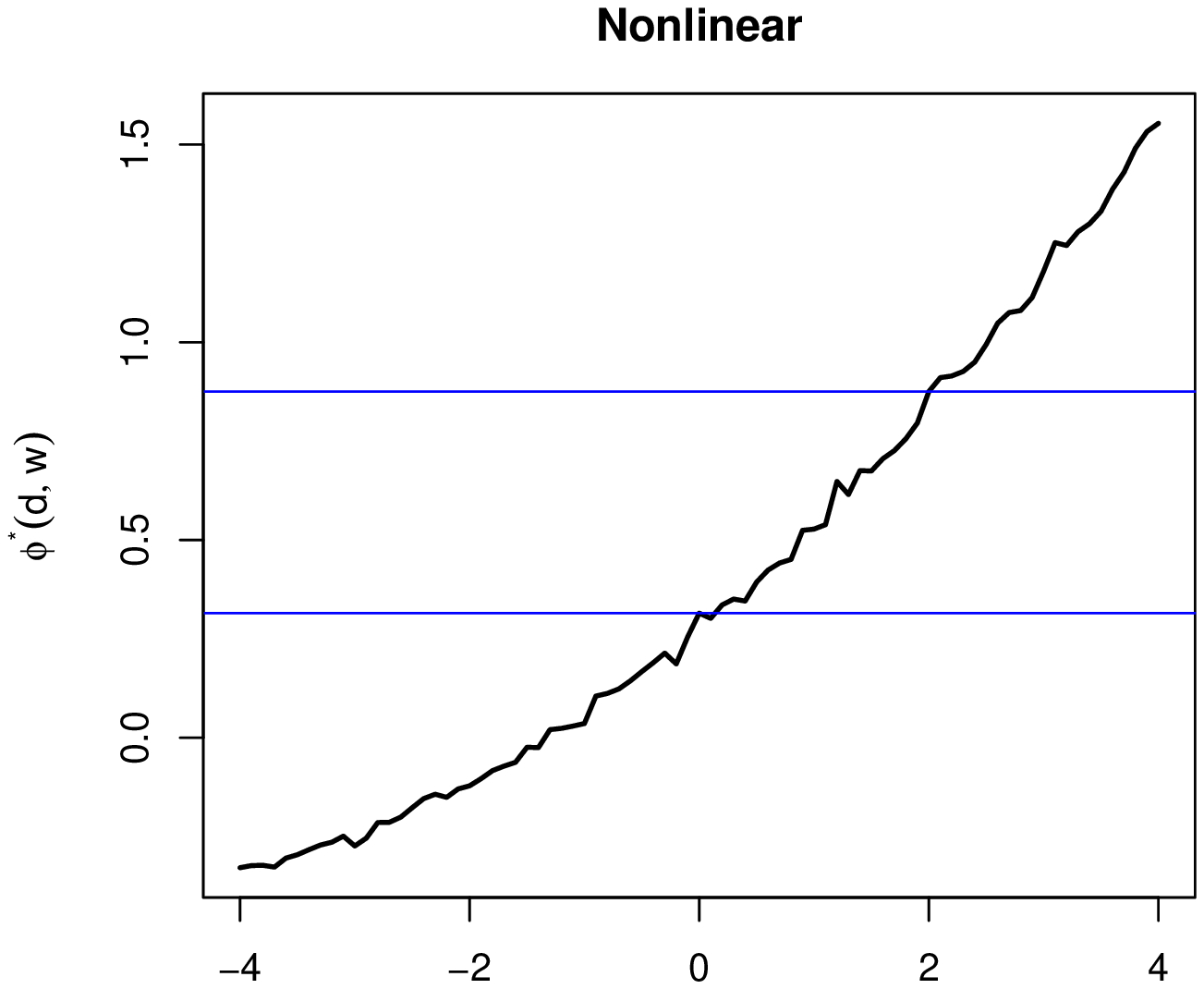}
\caption{The curves correspond to the functions $\phi^*(d,w)$ in scenarios (i) and (ii) considered in this section for $w=(0,\dots, 0,0.1)^{\intercal}\in \R^{7}$. The blue lines correspond to the true values for $d=-1$ and $d=2$ in each scenario and the difference between two blue lines is ${\rm CATE}(-1,2|w)$. }
\label{fig1} 
\end{figure}

\subsection{Binary outcome models}
\label{sec-simu-1}
For $1\leq i\leq n$, the exposure $d_i$ is generated as $d_i=z_i^{\intercal}\gam+v_i$ with $\gam = c_{\gam}\cdot (1,1,1,-1,-1,-1,-1)^{\intercal}$ and $v_i\sim N(0,1).$ We vary the IV strength such that $c_{\gam}\in\{0.4,0.6,0.8\}$. 
We generate two distributions of the $z_i$: (1) $\{z_i\}_{1\leq i\leq n}$ are \textit{i.i.d.} $N(0,{\rm I}_{p})$; (2) $\{z_i\}_{1\leq i\leq n}$ are \textit{i.i.d.} uniformly distributed in $[-1.73,1.73].$ For the model (i), we generate $\{y_i\}_{1\leq i\leq n}$ via the mixed-logistic model 
\begin{equation}
\label{eq: mixed logit}
\P(y_i=1 \mid d_i, w_i, u_i)=\text{logit}\left(d_i\beta+w_i^{\intercal}\kappa+u_i\right)
\end{equation} 
with $\beta=0.25$, $\kappa=\eta= (0,0,0,0,0,0.4,0.2)^{\intercal}.$ The unmeasured confounder $u_i$ is generated as $u_i=0.25v_i+w_i^{\intercal}\eta+\xi_i$ with $\xi_i\sim N(0,1)$.

After integrating out $u_i$ conditioning on $v_i,w_i$, the conditional distribution $y_i$ given $d_i, w_i$ is in general not logistic. Conditioning on $w_i$, the unmeasured confounder $u_i$ is correlated with $v_i$ and $d_i$. The majority rule is satisfied: the first five IVs are valid and the last two are invalid. 
We construct 95\% confidence intervals for ${\rm CATE}(-1,2|w)$ and compare our proposed SpotIV estimator with two other methods. The first one is the semi-parametric MLE assuming valid control function and valid IVs \citep{Rothe09}, shorthanded as Valid-CF. 
Through this comparison, we can understand how invalid IVs affect the accuracy of the causal inference approaches by assuming valid IVs. The second one is the ``Oracle" method, which is constructed with the prior knowledge of $\mathcal{V}$. It applies the Valid-CF by using the true valid IVs and treating the invalid IVs as known confounders. The oracle estimator is included as the benchmark. All the simulation results are calculated based on 500 replications.

In Table \ref{table1-ap}, we report the inference results for ${\rm CATE}(-1,2|w)$ for $w=(0,\dots, 0,0.1)^{\intercal}\in \R^{7}$ in the binary outcome model (i). The proposed SpotIV achieves the desired 95\% confidence level for Gaussian (Norm) and Uniform (Unif) $w_i$. The estimation errors get smaller with larger IV strengths and sample sizes. In contrast, the Valid-CF method has larger estimation errors, mainly due to the bias of using invalid IVs. The empirical coverage of the Valid-CF is lower than the nominal level across all settings. In terms of the empirical coverage, our proposed SpotIV is similar to the Oracle method, while SpotIV tends to have larger standard errors than the Oracle method. This happens since the SpotIV method identifies valid IVs in a data-dependent way. The ``MT'' column reports the proportion of simulations passing the majority rule test, detailed in Section \ref{sec3-5}. The majority rule is not rejected across most simulations. Our proposal is robust no matter whether the IVs are normal or uniform distributed.
\begin{table}[htbp]
\small
\begin{center}
\begin{tabular}{|c|c|c|cccc|ccc|ccc|}
\hline
&&&\multicolumn{4}{c|}{SpotIV}&\multicolumn{3}{c|}{Valid-CF}& \multicolumn{3}{c|}{ Oracle} \\
\hline
& $n$& $c_{\gam}$ & MAE & COV & SE &MT& MAE & COV & SE& MAE & COV & SE \\
 \hline
\multirow{6}{*}{Norm} &500 & 0.4  & 0.115&0.922& 0.14 &1& 0.205& 0.552& 0.11& 0.077& 0.914& 0.12 \\
 &500 & 0.6  & 0.082&0.936& 0.11&0.98& 0.142 & 0.590 & 0.08 & 0.067& 0.918& 0.09\\
& 500 & 0.8 & 0.070& 0.968& 0.10 &1& 0.130 & 0.592 & 0.07 & 0.055& 0.932& 0.08\\
\cline{2-13}
&1000 & 0.4 & 0.073& 0.918& 0.10&1& 0.204& 0.380& 0.09 & 0.055& 0.928& 0.08 \\
&1000 & 0.6 & 0.056& 0.920& 0.08&1& 0.160& 0.298& 0.06& 0.044& 0.944& 0.07\\
&1000 & 0.8 & 0.046& 0.960& 0.08&0.99& 0.131& 0.330& 0.05 & 0.040& 0.940& 0.06 \\
 \hline
\multirow{6}{*}{Unif} &500 & 0.4 & 0.112& 0.940& 0.14&1& 0.195& 0.638& 0.12 &  0.064& 0.966& 0.11 \\
&500 & 0.6  & 0.079& 0.964 &0.12&1& 0.165& 0.652& 0.10 & 0.054&0.968& 0.10  \\
&500 & 0.8 & 0.067& 0.968&0.10&1& 0.125& 0.752& 0.11 & 0.054 &0.984& 0.09\\
\cline{2-13}
&1000 & 0.4 & 0.082& 0.918& 0.11&1& 0.199 & 0.442 & 0.09 & 0.046& 0.958& 0.08 \\
&1000 & 0.6 & 0.052& 0.952& 0.09&0.99& 0.164 & 0.458 & 0.08& 0.042& 0.962& 0.07\\
&1000 & 0.8 & 0.052& 0.972& 0.08&0.99& 0.126& 0.616 & 0.08& 0.040& 0.986& 0.07 \\
 \hline
\end{tabular}
\end{center}
\caption{Inference for ${\rm CATE}(-1,2|w)$ in the binary outcome model (i). The ``MAE", ``COV" and ``SE" columns report the median absolute errors of $\widehat{\rm{CATE}}(-1,2|w)$, the empirical coverages of the confidence intervals and the average of estimated standard errors of the point estimators, respectively. The ``MT'' column reports the proportion of passing the majority rule testing in 500 replications.
The columns indexed with ``SpotIV" and ``Valid-CF" correspond to the proposed method and the method assuming valid IVs, respectively. The columns indexed with ``Oracle" correspond to the method which knows $\mathcal{V}$ as a priori.}
\label{table1-ap}
\end{table}


The identifiability condition considered in \citet{kolesar2015identification} and \citet{Bowden15} requires the IV strength vector $\gam$ and the invalidity form $\kappa+\eta$ to be nearly orthogonal in linear outcome models. 
Our configuration of $\gam$, $\kappa$, and $\eta$ in setting (i) corresponds to the case where this orthogonality assumption fails to hold, and our proposal is reliable. In the online supplement, we explore settings where the orthogonality assumption is satisfied. Our proposal is reliable regardless of whether this assumption holds or not, which matches our theoretical results.

\subsection{General nonlinear outcome models}
\label{sec-simu-2}
We consider a nonlinear continuous outcome model as follows. 
\begin{itemize}
\item[(ii)] Generate $\{y_i\}_{1\leq i\leq n}$ via
$
 y_i=d_i\beta+z_i^{\intercal}\kappa+u_i+(d_i\beta+z_i^{\intercal}\kappa+u_i)^2/3.
$ 

\end{itemize}
The true parameters and the distribution of $u_i$ in (ii) are set to be the same as in (i) in Section \ref{sec-simu-1}.

We compare the SpotIV estimator with the two-stage hard-thresholding (TSHT) method \citep{TSTH}, which is proposed to deal with possibly invalid IVs in linear outcome models. This comparison aims to understand the effect of mis-specifying a nonlinear model as linear.
As reported in Table \ref{table3-ap}, the proposed SpotIV method has coverage probabilities close to 95\% in model (ii). In comparison, the TSHT does not guarantee the 95\% coverage and has larger estimation errors, mainly due to the fact that the TSHT method is only developed for linear outcome models. 
\begin{table}[htp!]
\small
\begin{center}
\begin{tabular}{|c|c|c|cccc|ccc|ccc|}
\hline
&&&\multicolumn{4}{c|}{SpotIV}&\multicolumn{3}{c|}{TSHT}& \multicolumn{3}{c|}{ Oracle} \\
\hline
& $n$& $c_{\gam}$ & MAE & COV & SE &MT& MAE & COV & SE& MAE & COV & SE \\
 \hline
\multirow{6}{*}{Norm} &500 & 0.4  & 0.084&0.932& 0.12 &1& 1.693& 0.040& 0.27 & 0.058& 0.960& 0.09 \\
 &500 & 0.6  & 0.070&0.916& 0.07&0.98& 1.169& 0.093& 0.18 & 0.045& 0.978& 0.07\\
& 500 & 0.8 &0.055&0.928& 0.08&1& 0.865& 0.110& 0.13 & 0.046& 0.972& 0.07\\
\cline{2-13}
&1000 & 0.4 & 0.059& 0.910& 0.09&1& 1.535& 0.107& 0.21 & 0.045& 0.920& 0.06 \\
&1000 & 0.6 & 0.046& 0.938& 0.07&1& 0.399& 0.363& 0.13& 0.036& 0.952& 0.05\\
&1000 & 0.8 & 0.037& 0.950& 0.06&1& 0.268& 0.383& 0.10 & 0.030& 0.966& 0.05 \\
 \hline
\multirow{6}{*}{Unif} &500 & 0.4 & 0.289& 0.910& 0.44&1& 0.847& 0.093& 0.13 &  0.200& 0.949& 0.36 \\
&500 & 0.6  & 0.243& 0.918 &0.35&1& 1.129& 0.067& 0.18 & 0.170&0.974& 0.30 \\
&500 & 0.8 & 0.187& 0.936&0.30&1& 0.847& 0.093& 0.13 & 0.147 &0.952& 0.28\\
\cline{2-13}
&1000 & 0.4 & 0.199& 0.892& 0.30&1& 0.583 & 0.140 & 0.20 & 0.138& 0.944& 0.25 \\
&1000 & 0.6 & 0.147& 0.948& 0.23&1& 0.323 & 0.133 & 0.05& 0.113& 0.956& 0.21\\
&1000 & 0.8 & 0.130& 0.938& 0.21&1& 0.267& 0.423 & 0.10& 0.106& 0.956& 0.18\\
 \hline
\end{tabular}
\end{center}
\caption{Inference for ${\rm CATE}(-1,2|w)$ in the nonlinear outcome model (ii). The ``MAE", ``COV" and ``SE" columns report the median absolute errors of $\widehat{\rm{CATE}}(-1,2|w)$, the empirical coverages of the confidence intervals and the average of estimated standard errors of the point estimators, respectively. The ``MT'' column reports the proportion of passing the majority rule testing in 500 replications. The columns indexed with ``SpotIV" and ``TSHT" correspond to the proposed method and the method in \citet{TSTH}, respectively. The columns indexed with ``Oracle" correspond to the method which knows $\mathcal{V}$ as a priori.}
\label{table3-ap}
\end{table}

\subsection{Violation of majority rule}
\label{ap-simu-3}
In this subsection, we examine the performance of our proposal when the majority rule is violated. Specifically, we consider the outcome model (i) with the two ways of violating the majority rule: (a) $\kappa=\eta=(0.4,0.4,0.4,0,0.4,0.4,0.4)^{\intercal}$; (b) $\kappa_j=\eta_j=\tilde{\xi}_j\cdot \gam_j$ for $1\leq j\leq 7$ with $\tilde{\xi}_j\sim U[-1,1]$. 
 The parameter $\gam$ is the same as in Section \ref{sec-simu-1}. In (a), $\mathcal{V}=\{4\}$ and 
$\kappa_j/\gam_j=\eta_j/\gam_j=0.4$ for $j=1,2,3$ and $\kappa_j/\gam_j=\eta_j/\gam_j=-0.4$ for $j=5,6,7$.
Hence, no IV should get four or more votes and the violation setting (a) is likely to be detected by the voting method detailed in Section \ref{sec3-5}. In (b), $\mathcal{V}=\emptyset$ and $ \kappa_j/\gam_j \sim U[-1,1]$ and hence no IV should get four or more votes. In the setting (b), $\kappa_j/\gam_j$ is not as spread out as in (a) and hence our proposed voting method may be less powerful in comparison to the setting (a).
 In Table \ref{table1-violate}, we present the inference results under the configurations (a) and (b) with uniformly distributed IV measurements. 
In both settings, the SpotIV and Valid-CF have low coverages, which are as expected since the majority rule is violated. In setting (a), the violation of the majority rule is detected in most simulations. In setting (b), we cannot detect the violation of the majority rule in many cases. However, the detection gets easier as the IVs get stronger. 

\begin{table}[!htbp]
\small
\begin{center}
\begin{tabular}{|c|c|c|cccc|ccc|}
\hline
&&&\multicolumn{4}{c|}{SpotIV}&\multicolumn{3}{c|}{Valid-CF} \\
\hline
& $n$& $c_{\gam}$ & MAE & COV & SE &MT& MAE & COV & SE\\
 \hline
\multirow{3}{*}{(a)} &1000 & 0.4 &0.097 & 0.910 & 0.13 & 0.02 & 0.140 & 0.657 & 0.09\\
&1000 & 0.6 & 0.071 & 0.937 &0.09 &0.02 & 0.116 & 0.573 &0.06\\
&1000 & 0.8 & 0.062 & 0.897 & 0.07 & 0.02 & 0.105& 0.543 & 0.05\\
 \hline
 \multirow{3}{*}{(b)} &1000 & 0.4 &0.136 & 0.743 & 0.11 & 0.80 & 0.104 & 0.708 & 0.08\\
&1000 & 0.6 & 0.140 & 0.567 &0.09 &0.71 & 0.091& 0.644 &0.06\\
&1000 & 0.8 & 0.137 & 0.523 & 0.07 & 0.66& 0.083 & 0.610 & 0.05\\
 \hline
\end{tabular}
\end{center}
\caption{Inference for ${\rm CATE}(-1,2|w)$ in the binary outcome model (i) when the majority rule is violated. The columns indexed with ``MAE", ``COV" and ``SE" report the median absolute errors of $\widehat{\rm{CATE}}(-1,2|w)$, the empirical coverages of the confidence intervals and the average of estimated standard errors of the point estimators, respectively. The ``MT'' column reports the proportion of passing the majority rule testing in 500 replications. The columns indexed with ``SpotIV" and ``Valid-CF" correspond to the proposed method and the method assuming valid IVs, respectively.}
\label{table1-violate}
\end{table}

\section{Real data analysis}
\label{sec-data}
We apply the proposed SpotIV method to infer the causal effect of the income on the house owning through analyzing the data from China Family Panel Studies \citep{xie2012china}. The outcome $Y$ is a dichotomous variable indicating whether {a family owns a house (1) or not (0)}. 
The exposure $D$ is the log-transformed family income per person as the income variable is highly right-skewed. The baseline covariates $X$ include the {age}, {gender}, and {marriage status of the head of household}. We include seven candidate IVs: the number of books at home, the education level of the head of the household, the registered residence type of the head of the household, monthly fees on dining, transport, travel, and {health}.
In the preprocessing step, missing values are removed which gives a sample with the response ratio ``own'': ``not own''$\approx 5.88:1$. As the response ratio is highly unbalanced, we randomly remove part of the samples with positive response and arrive at a sample where ``own'': ``not own''$=2:1$. This gives a sample with 2538 observations. 
 We choose these candidate IVs following a similar rationale as in the data analysis in \citet{Rothe09}. Specifically, the first three candidate IVs are measures of human capital, which can be strongly related to income but have little effect on the housing decision. On the other hand, home ownership could be determined by the permanent components of the income stream. The expenses on different aspects (the last four candidate IVs) respond to the income \citep{campbell1991response}. Consumption in daily activities is a necessary part of life and always exists but do not decide on house investments. Therefore, the  monthly consumptions are potentially valid IVs for the current study.

Applying the SpotIV method, the number of books at home and monthly fees on health are excluded from $\widehat{\mathcal{S}}$ and the remaining five IVs are selected as relevant IVs. Among these five IVs, the SpotIV method chooses {the registered residence type of the head of the household}, {the monthly fees on transport}, and the monthly fees on travel as valid IVs. The SpotIV method chooses the monthly fees on dining and monthly fees on health as invalid IVs.

In Figure \ref{fig-wage}, we report the inference results for the conditional average treatment effect for the male and female groups, respectively. Before adjusting for the endogeneity, income and house ownership tend to be negatively associated even though the relationship is insignificant. This result is not reliable due to the existence of unmeasured confounders. 
Assuming all IVs to be valid, the estimated causal effects are close to zero. This result can also be suspicious as two IVs are detected as invalid IVs based on the SpotIV method. Our proposed estimator leads to a significant and positive causal effect of income on house ownership. Comparing the fitted curves for the male and female groups, we observe an interesting phenomenon that the causal effects for males tend to be larger than those for females among subjects with a high or low income. 
\begin{figure}[H]
\centering
\includegraphics[height=5.8cm,width=0.45\textwidth]{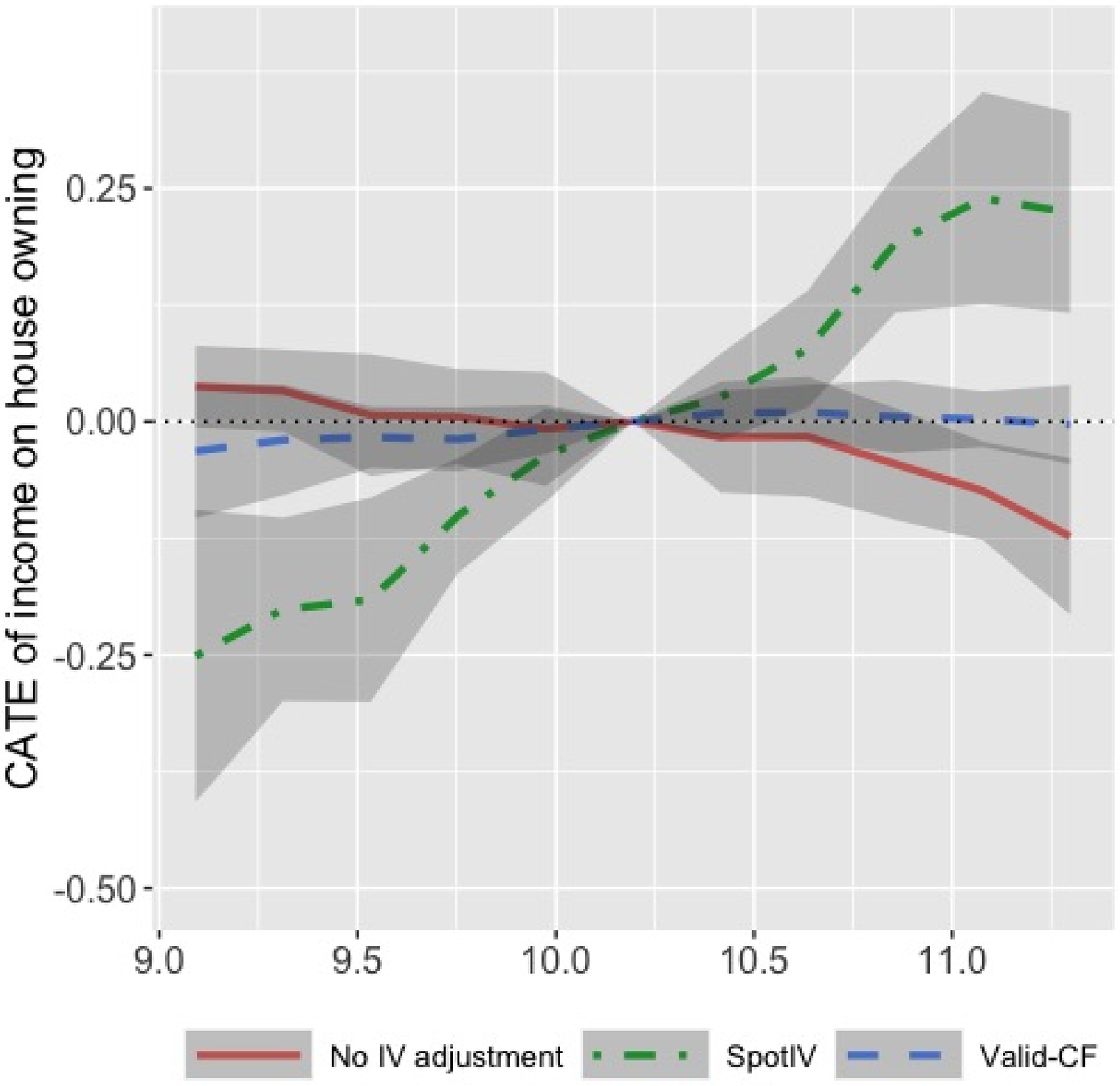}
\includegraphics[height=5.8cm,width=0.45\textwidth]{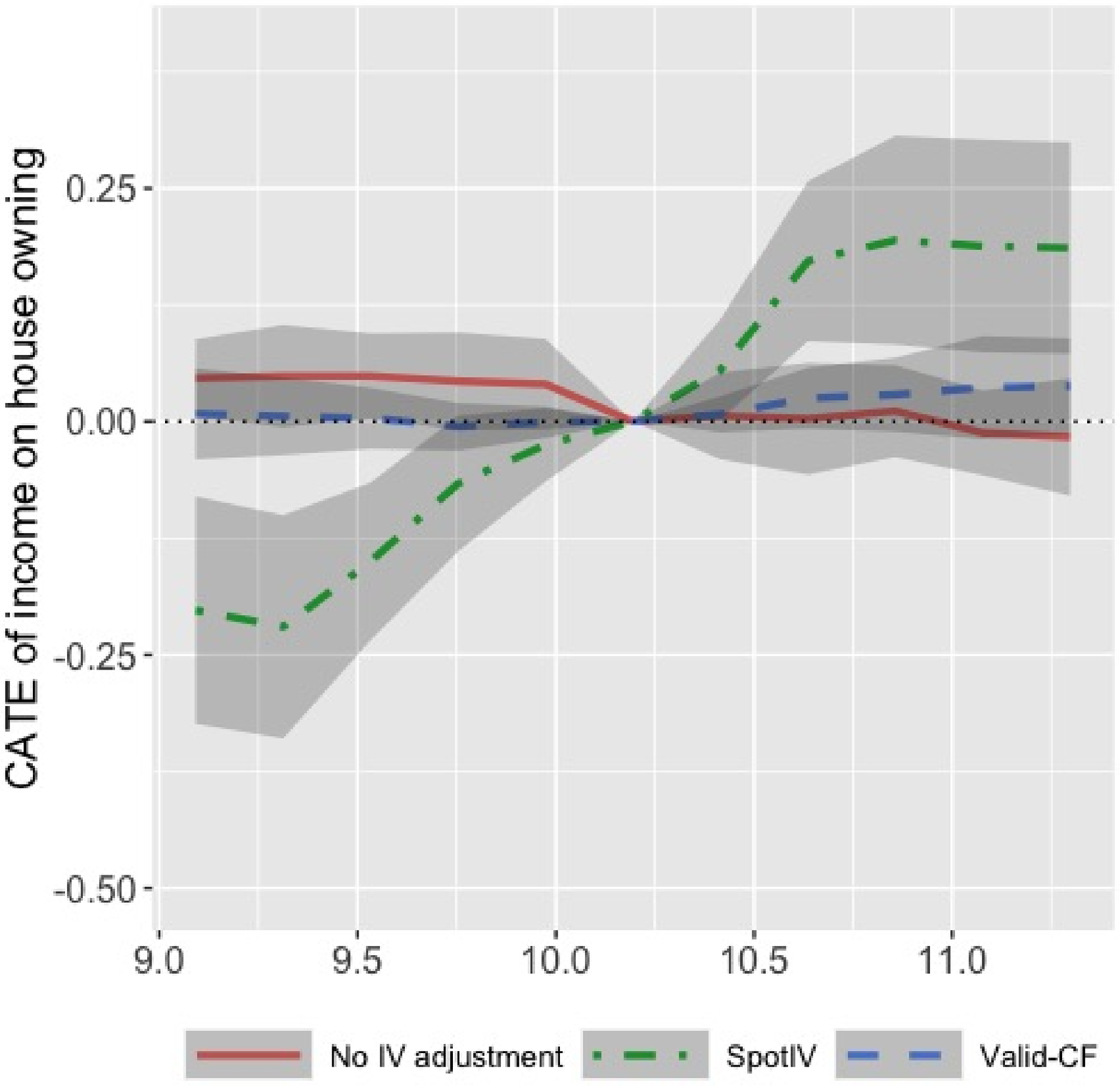}
\caption{The $x$-axis denotes log-transformed income $(d_1)$ and the $y$-axis denotes CATE($d_1,10.2|w_M$) (left) and CATE($d_1,10.2|w_F$) (right), where $w_M$ and $w_F$ are the average values of observed covariates for males and females respectively and the value 10.2 on the $x$-axis is the mean of log-transformed income in the sample. The solid lines are estimated without accounting for unmeasured confounders; the dot-dash lines are estimated based on the proposed SpotIV; the dashed lines are estimated assuming valid control functions (Condition \ref{cond: valid IC}). The shaded area corresponds to the point-wise 95\% confidence intervals for each estimate.}
\label{fig-wage}
\end{figure}

\section*{Acknowledgement}
The research of S. Li was supported by the Fundamental Research Funds for the Central Universities, and the Research Funds of Renmin University of China. 
The research of Z. Guo was supported in part by the NSF grants DMS-1811857, DMS-2015373 and NIH-1R01GM140463-01. We thank Junhui Yang for her help with cleaning the data used in the empirical study. 


\appendix
\section{Proof of Lemma \ref{lem2-semi}}
\label{ap-sec2}
\begin{proposition}
\label{prop-ap}
Under Condition \ref{cond-semi1},
$\E[y_i|w_i,v_i]=\E[y_i|w_i^{\intercal}\Theta,v_i]$.
\end{proposition}
 \begin{proof}[Proof of Proposition \ref{prop-ap}]
Let $\bar{\Sig}=\E(w_i^{\intercal},v_i)(w_i^{\intercal},v_i)^{\intercal}]$, $\bar{\alpha}(y_i)=\E[\bar{\Sig}^{-1/2}(w_i^{\intercal},v_i)|y_i]\in \R^{p+1}$, and $\bar{\Omega}=\textup{cov}(\bar{\alpha}(y_i))$.
We first show that $\mathcal{L}(\bar{\Omega})\subseteq \mathcal{L}\left(\begin{pmatrix}
\Theta^*& 0\\
0 & 1\end{pmatrix}\right)$. Next, we show \begin{align}
\label{sir-eq1}
\mathcal{L}(\bar{\Omega}_{1:p,1:p})=\mathcal{L}(\Omega)= \mathcal{L}\left(\Theta^*\right).
\end{align}

We know that for binary outcomes,
\begin{align*}
\E[y_i|w_i,v_i]=\P(y_i=1|w_i,v_i)=\P(y_i=1|w_i^{\intercal}\Theta^*,v_i).
\end{align*}
That is, $y_i\perp (w_i,v_i)|(w_i^{\intercal}\Theta^*,v_i)$. 
By Condition \ref{cond-semi1}, $\E[(w_i^{\intercal},v_i)|w_i^{\intercal}\Theta^*,v_i]$ is linear in $w_i^{\intercal}\Theta^*,v_i$. Therefore, by Theorem 3.1 in \citet{Li91}, we know that $\mathcal{L}(\bar{\Omega})\subseteq \mathcal{L}\left(\begin{pmatrix}
\Theta^*& 0\\
0 & 1\end{pmatrix}\right)$.

Next, we show (\ref{sir-eq1}). The first equality holds because $\E[w_i^{\intercal}v_i]=0$ by (\ref{eq-d}). As $\mathcal{L}(\bar{\Omega})\subseteq \mathcal{L}\left(\begin{pmatrix}
\Theta^*& 0\\
0 & 1\end{pmatrix}\right)$, we know 
\begin{align}
\label{eq-Theta-trans}
\bar{\Omega}=\begin{pmatrix}
\Theta^*& 0\\
0 & 1\end{pmatrix}\begin{pmatrix}
R_{1,1} & R_{1,2}\\
R_{2,1} & R_{2,2}\end{pmatrix}
\end{align}
for some constant matrix $R\in\R^{(p_{\eta}+2)\times (p+1)}$ and $R_{1,1}\in\R^{(p_{\eta}+1)\times p}$.
We arrive at $\bar{\Omega}_{1:p,1:p}=\Theta^* R_{1,1}$. As rank$(\bar{\Omega}_{1:p,1:p})=M=p_{\eta}+1=$rank$(\Theta^*)$, the proof of (\ref{sir-eq1}) is complete now.

  \end{proof}

 \begin{proposition}[Convergence rate of $\widehat{\Theta}$]
 \label{prop1}
 Under the conditions of Lemma \ref{lem2-semi}, we have
 \[
  \P\left( \max_{1\leq m\leq M}\|\widehat{\Theta}_{.,m}-\Theta_{.,m}\|_2\geq C_1\sqrt{\frac{t}{n}}\right)\leq \exp(-C_2t).
\]
 \end{proposition}
\begin{proof}[Proof of Proposition\ref{prop1}]
Notice that
\begin{align*}
\Omega=\Sig^{-1/2}Cov(\alpha(y_i))\Sig^{-1/2}=\Sig^{-1/2}\E[\alpha(y_i)\alpha(y_i)^{\intercal}]\Sig^{-1/2}
\end{align*}
as $\E[\alpha(y_i)]=\E[w_i]=0$.

The following decomposition holds
\begin{align}
\label{decom1}
\|\widehat{\Omega}-\Omega\|_2&\leq 2\|\Sig^{-1/2}-\widehat{\Sig}^{-1/2}\|_2\|cov(\alpha(y_i))\Sig^{-1/2}\|_2\nonumber\\
&\quad +\|\Sig^{-1/2}\|^2_2\|cov(\alpha(y_i))-\frac{1}{n}\sum_{i=1}^n\hat{\alpha}(y_i)\hat{\alpha}(y_i)^{\intercal}\|_2+r_n,
\end{align}
where $r_n$ is of smaller order than the first two terms.

For the first term,
\begin{align*}
\|\Sig^{-1/2}-\widehat{\Sig}^{-1/2}\|_2
\leq\|\Sig-\widehat{\Sig}\|_2\|\Sig^{1/2}+\widehat{\Sig}^{1/2}\|^{-1}_2.
\end{align*}
Since $\widehat{\Sig}$ is an average of \textit{i.i.d.} sub-exponential variables, we have
\[
\P\left( \|\Sig-\widehat{\Sig}\|_2\geq c\sqrt{t/n}\right)\leq \exp(-ct).
\]
As $\|\Sig^{-1/2}cov(\alpha(y_i))\|_2\leq C<\infty$, for the first term in (\ref{decom1}),
\begin{align}
\label{eq0-pf1}
\P\left( 2\|\Sig^{-1/2}-\widehat{\Sig}^{-1/2}\|_2\|cov(\alpha(y_i))\Sig^{-1/2}\|_2\geq c_1\sqrt{t/n}\right)\leq \exp(-c_2t).
\end{align}

To bound the second term in (\ref{decom1}), for  binary $y_i$, it holds that
\begin{align*}
\alpha(1)=\E[w_i|y_i=1]~~\hat{\alpha}(1)=\frac{1}{\sum_{i=1}^n\mathbbm{1}(y_i=0)}\sum_{i=1}^nw_i\mathbbm{1}(y_i=1)\\
\alpha(0)=\E[w_i|y_i=0]~~\hat{\alpha}(0)=\frac{1}{\sum_{i=1}^n\mathbbm{1}(y_i=0)}\sum_{i=1}^nw_i\mathbbm{1}(y_i=0).
\end{align*}
By some simple algebra, we can show that
\begin{align*}
&cov(\alpha(y_i))=\P(y_i=1)\P(y_i=0)(\alpha(1)-\alpha(0))(\alpha(1)-\alpha(0))^{\intercal}.
\end{align*}

The following decomposition holds
\begin{align*}
&\left\|\frac{1}{n}\sum_{i=1}^n\hat{\alpha}(y_i)\hat{\alpha}(y_i)^{\intercal}-\frac{cov(\alpha(y_i))}{\P(y_i=1)\P(y_i=0)}\right\|_2\\
&\leq 2\|(\hat{\alpha}(1)-\hat{\alpha}(0)-\alpha(1)+\alpha(0))(\alpha(1)-\alpha(0))^{\intercal}\|_2+\|\alpha(1)-\alpha(0)-\hat{\alpha}(1)+\hat{\alpha}(0)\|_2^2\\
&\leq 4\|\alpha(1)-\alpha(0)\|_2\max_{k\in \{0,1\}}\|\hat{\alpha}(k)-\alpha(k)\|_2+4\max_{k\in \{0,1\}}\|\hat{\alpha}(k)-\alpha(k)\|_2^2.
\end{align*}
First notice that $\alpha(k)=\frac{\E\left[w_i\mathbbm{1}(y_i=k)\right]}{\P(y_i=k)}$. Therefore,
\begin{align*}
\|\hat{\alpha}(k)-\alpha(k)\|_2&\leq |\frac{1}{\P(y_i=k)}-\frac{n}{\sum_{i=1}^n\mathbbm{1}(y_i=k)}|\|\E\left[w^{\intercal}_i\mathbbm{1}(y_i=k)\right]\|_2\\
&+\frac{1}{\P(y_i=k)}\|\frac{1}{n}\sum_{i=1}^n w_i\mathbbm{1}(y_i=k)-\E\left[w_i\mathbbm{1}(y_i=k)\right]\|_2.\\
   |\frac{1}{n}\sum_{i=1}^n(\hat{v}_i-v_i)\mathbbm{1}(y_i=k)|&= |\frac{1}{n}\sum_{i=1}^n\mathbbm{1}(y_i=k)w_i^{\intercal}(\hat{\gam}-\gam)|\\
   &\leq  \|\frac{1}{n}\sum_{i=1}^n\mathbbm{1}(y_i=k)w_i^{\intercal}\|_2\|\hat{\gam}-\gam\|_2.
\end{align*}
Hence, it is straight forward to show that
\[
  \P\left(\left\|\frac{1}{n}\sum_{i=1}^n\hat{\alpha}(y_i)\hat{\alpha}(y_i)^{\intercal}-\frac{cov(\alpha(y_i))}{\P(y_i=1)\P(y_i=0)}\right\|_2\geq c_3\sqrt{t/n}\right)\leq \exp(-c_4t)
\]
for sufficiently large constants $c_3$ and $c_4$.

In view of (\ref{decom1}), we have shown
\begin{equation}
\label{pf-eq1}
   \P\left(\|\widehat{\Omega}-\Omega\|_2\geq c_5\sqrt{\frac{t}{n}}\right)\leq \exp(-c_6t)
\end{equation}
for sufficiently large constants $c_5$ and $c_6$.

Next, we show the the eigenvalues of $\widehat{\Omega}$ converge to the eigenvalues of $\Omega$.
In fact,
\begin{align*}
&\max_{1\leq k\leq p} \left|\hat{\lambda}_k- \hat{\lambda}_k\right|\leq \max_{\|u\|_2=1}|u^{\intercal}(\widehat{\Omega}-\Omega)u|\leq \|\widehat{\Omega}-\Omega\|_2.
\end{align*}
For the eigenvectors,  we use Theorem 5 of \citet{Karoui08}. Under Condition \ref{cond-semi1}, we have
\begin{align*}
\|\widehat{\Phi}_{.,m}-\Phi_{.,m}\|_2\leq \frac{\|\widehat{\Omega}-\Omega\|_2}{\lambda_m(\Omega)}~~\forall ~1\leq m\leq 2.
\end{align*}
In view of (\ref{pf-eq1}), we have shown
\[
  \P\left( \max_{1\leq m\leq M}\|\widehat{\Theta}_{.,m}-\Theta_{.,m}\|_2\geq C_1\sqrt{\frac{t}{n}}\right)\leq \exp(-C_2t).
\]

\end{proof}

\begin{proof}[Proof of Lemma \ref{lem2-semi}]
For $\widehat{\gam}$ computed via (\ref{est1}), under Condition \ref{cond-semi}, it is easy to show that
\begin{equation}
\label{eq0-pf2}
  \sqrt{n}(\hat{\gam}-\gam)\xrightarrow{D} N\left(0,\sig^2_v\E^{-1}[w_iw_i^{\intercal}]\right).
\end{equation}

Define an event
\begin{align}
\label{eq-E1}
E_1=\left\{\widehat{M}=M,\widehat{\mathcal{S}}=\mathcal{S} \right\}.
\end{align}
We first show that $\P(E_1)\rightarrow 1$ as $n\rightarrow \infty$.
Given the results of Proposition \ref{prop1}, by Theorem 2 of \citet{Zhu06}, we know that for $C_n=n^{c_0}$ and $c_0\in(0,1)$, $\P(\widehat{M}=M)\rightarrow 1$.

For the last statement in $E_1$, it is easy to show 
\[ 
  |\hat{\sig}_v^2-\sig_v^2|=O_P(n^{-1/2}).
\]
Let $\widehat{\omega}_j=\hat{\sig}^2_v \{\widehat{\Sig}^{-1}\}_{j,j}$. 
For $j \in \mathcal{S}$, we have
\begin{align*}
&\P\left(\left|\widehat{\gam}_j\right|\geq\sqrt{\widehat{\omega}_{j}} \sqrt{\frac{2.01 \log n }{n}}\right)\geq \P\left(|\gam_j|-\left|\widehat{\gam}_j-\gam_j\right|\geq \sqrt{\widehat{\omega}_{j}} \sqrt{\frac{2.01 \log n }{n}}\right)\\
&=\P\left(\left|\widehat{\gam}_j-\gam_j\right|\leq |\gam_j|-\sqrt{\widehat{\omega}_{j}} \sqrt{\frac{2.01 \log n }{n}}\right)\rightarrow 1,
\end{align*}
where the convergence follows from (\ref{eq0-pf2}) and $|\gam_j|\geq c_0>0$ for $j\in\mathcal{S}$.
For $j \in \mathcal{S}^{c}$, we have
\begin{align*}
&\P\left(\left|\widehat{\gamma}_j\right|>\sqrt{\widehat{\omega}_{j}}\sqrt{\frac{2.01 \log n }{n}} \right)=\P\left(\left|\widehat{\gamma}_j-\gam_j\right|>\sqrt{\widehat{\omega}_{j}}\sqrt{\frac{2.01 \log n }{n}} \right)=o(1),
\end{align*}
where the last step is due to $\|\hat{\gam}-\gam\|_2=O_P(n^{-1/2})$.
Combining above two expressions, we have establish that
\begin{equation}
\mathbf{P}\left(\widehat{\mathcal{S}}=\mathcal{S}\right)\rightarrow 1.
\label{eq: selection case 0}
\end{equation}

It suffices to prove the rest of the results conditioning on the event $E_1$. Under the majority rule, the median of $\{\widehat{\Theta}_{j,m}/\hat{\gam}_j\}_{j\in S}$ must be evaluated at a valid IV  for $m=1,\dots,M$.

Together with Proposition \ref{prop1}, we have
\begin{align}
\label{eq2-lem1}
\P\left(\max_{1\leq m\leq M}\max_{j\in\mathcal{S}}\left|\frac{\widehat{\Theta}_{j,m}}{\hat{\gam}_j}-\frac{\Theta_{j,m}}{\gam_j}\right|\geq c_1\sqrt{\frac{t}{n}}\right)\leq \exp(-c_2t)
\end{align}
for some positive constants $c_1$ and $c_2$. Notice that
\[
\max_{1\leq m\leq M}|\hat{b}_m-b_m|\leq \max_{1\leq m\leq M}\max_{j\in\mathcal{S}}\left|\frac{\widehat{\Theta}_{j,m}}{\hat{\gam}_j}-\frac{\Theta_{j,m}}{\gam_j}\right|.
\]
The proof for (\ref{eq: B-accuracy}) is complete now.

Finally, we are left to show $\E[y_i|d_i,w_i^{\intercal},v_i]=\E[y_i|(d_i,w_i^{\intercal})B,v_i]$. 
To see this, notice that by Condition \ref{cond-semi1},
\[
  \E[y_i|w_i,v_i]=\E[y_i|w_i^{\intercal}\Theta^*,v_i]=\E[y_i|w_i^{\intercal}\Theta,v_i].
\]
As $\Theta$ has the smallest possible dimension, we know $\Theta=\Theta^*T$ for some constant matrix $T\in\R^{p_{\eta}+1}\times M$.
Therefore,
\begin{align*}
\E[y_i|d_i,w_i,v_i]&=\E[y_i|w_i,v_i]=\E[y_i|w_i^{\intercal}\Theta,v_i]=\E[y_i|w_i^{\intercal}\Theta^*T,v_i]\\
&=\E[y_i|(d_i-v_i,w_i^{\intercal})B^*T,v_i]\\
&=\E[y_i|(d_i,w_i^{\intercal})B^*T,v_i]
\end{align*}

To complete the proof, we only need to show $B$ defined in (\ref{eq-B-ap}) satisfies that $B=B^*T$ given that $\Theta=\Theta^*T$ for some invertible $T\in\R^{(p_{\eta}+1)\times (p_{\eta}+1)}$ by (\ref{eq-Theta-trans}).
Notice that $\Theta_{j,.}=\Theta^*_{j,.}T$. If $j\in V$, then $\Theta_{j,.}=(\gam_j \beta, \bm{0}_{p_{\eta}}^{\intercal})T$ and $\Theta_{j,.}/\gam_j=(\beta,\bm{0}_{p_{\eta}}^{\intercal})T$.
Therefore, $b_k=Median(\{\Theta_{j,k}/\gam_j\}_{j\in\mathcal{S}})=(\beta,\bm{0}_{p_{\eta}}^{\intercal})T_{.,k}$ for $k=1,\dots, M$. Moreover,
\[
  \Theta_{.,k}-b_k\gam=(\gam \beta+\kappa, \eta)T_{.,k}-(\beta\gam,\bm{0}_{p_{\eta}}^{\intercal})T_{.,k}=(\kappa,\eta)T_{.,k},~k=1,\dots,M.
\]
Hence, we have showed that $B=B^*T$.
  \end{proof}

\section{Proof of Theorem \ref{thm: asf}}
\label{ap-sec4}

It follows from the condition $h=n^{-c}$ for $0<c<1/4$ that $nh^4\gg \log n$ and $h\log n\rightarrow 0.$
We recall the following definitions,
\begin{equation*}
t_i=((d_i,w_i^{\intercal})B,{v}_i)^{\intercal},\quad \widehat{t}_i=((d_i,w_i^{\intercal})\widehat{B},\widehat{v}_i)^{\intercal}, \quad {s}_i=((d,w^{\intercal})B,{v}_i)^{\intercal},\quad \widehat{s}_i=((d,w^{\intercal})\widehat{B},\widehat{v}_i)^{\intercal}.
\end{equation*}

In event $E_1$ (\ref{eq-E1}), $\{\widehat{M}=M=2\}$ and the kernel is defined in three dimensions. That is, for $a,b\in \R^{3},$
$$
K_{H}(a, b)=\prod_{l=1}^{3} \frac{1}{h}k\left(\frac{a_{l}-b_{l}}{h}\right),
$$
where $h$ is the bandwidth and $k(x)={\bf 1}\left(|x|\leq 1/2\right).$ 
We define the events
$$ \mathcal{A}_1=\left\{\|\widehat{B}-B\|_2\leq C \sqrt{\frac{\log n}{n}}, \|\widehat{\gamma}-\gamma\|_2\leq C \sqrt{\frac{\log n}{n}}\right\},~~\mathcal{A}_2=\max\{\|w_i\|_{\infty}, |d_i|\}\lesssim \sqrt{\log n}.$$
By Lemma \ref{lem2-semi} and $w_i$ and $v_i$ being sub-gaussian, we establish that {$\mathbf{P}(\mathcal{A}_1\cap \mathcal{A}_3)\geq {\hprob}-P(E_1)$.}  On the event $\mathcal{A}_1\cap \mathcal{A}_2,$ we have
$$
\max_{1\leq i\leq n} \max\left\{ \|\widehat{s}_i-s_i\|_2,\|\widehat{t}_i-t_i\|_2 \right\}\leq C {{\log n}}/{\sqrt{n}}
$$
for a large positive constant $C>0.$ 

We start with the decomposition 
\begin{equation}
\widehat{\phi}(d,w)-{\phi}^*(d,w)
=\frac{1}{n} \sum_{i=1}^{n} \left[\widehat{g}(\widehat{s}_i)-{g}(\widehat{s}_i)\right]+\frac{1}{n} \sum_{i=1}^{n} {g}(\widehat{s}_i) -\int {g}({s}_i) f_{v}(v_i) dv_i
\label{eq: error decompotision}
\end{equation}
where $f_{v}$ is the density of $v_i.$
By Lemma \ref{lem2-semi}, we define 
\begin{equation}
\epsilon_i=y_i-\E[y_i|(d_i,w^{\intercal}_i)B,v_i]=y_i-g((d_i,w_i^{\intercal})B,v_i) \quad \text{for} \; 1\leq i\leq n.
\label{eq: reg error definition}
\end{equation}
We plug in the expression of $\widehat{g}(\widehat{s}_i)$ and decompose the error $\frac{1}{n} \sum_{i=1}^{n} \left[\widehat{g}(\widehat{s}_i)-{g}(\widehat{s}_i)\right]$ as 
\begin{equation}
\begin{aligned}
&\frac{1}{n} \sum_{i=1}^{n} \frac{\sum_{j=1}^{n}[y_j-{g}(\widehat{s}_i)] K_{H}(\widehat{s}_i, \widehat{t}_j)}{\sum_{j=1}^{n} K_{H}(\widehat{s}_i, \widehat{t}_j)}
=\frac{1}{n} \sum_{i=1}^{n} \frac{\sum_{j=1}^{n}\epsilon_j  K_{H}(\widehat{s}_i, \widehat{t}_j)}{\sum_{j=1}^{n} K_{H}(\widehat{s}_i, \widehat{t}_j)}\\
&+\frac{1}{n} \sum_{i=1}^{n} \frac{\sum_{j=1}^{n}[g(\widehat{t}_j)-{g}(\widehat{s}_i)]  K_{H}(\widehat{s}_i, \widehat{t}_j)}{\sum_{j=1}^{n} K_{H}(\widehat{s}_i, \widehat{t}_j)}
+\frac{1}{n} \sum_{i=1}^{n} \frac{\sum_{j=1}^{n}[g({t}_j)-{g}(\widehat{t}_j)]  K_{H}(\widehat{s}_i, \widehat{t}_j)}{\sum_{j=1}^{n} K_{H}(\widehat{s}_i, \widehat{t}_j)}.
\end{aligned}
\label{eq: second decomp}
\end{equation}
Since 
$$\left|{g}(\widehat{t}_j)-g({t}_j)\right|\cdot K_{H}(\widehat{s}_i, \widehat{t}_j)\leq \|\dg({t}_j+c(\widehat{t}_j-{t}_j))\|_2 \|\widehat{t}_j-{t}_j\|_2 \cdot K_{H}(\widehat{s}_i, \widehat{t}_j),$$ {we apply the boundedness assumption on $\dg$ imposed in Condition 5.3 (b) and obtain that 
$\left|{g}(\widehat{t}_j)-g({t}_j)\right|\lesssim \log n/\sqrt{n}$ on the event $\mathcal{A}.$} Here, we use the fact that, if $K_{H}(\widehat{s}_i, \widehat{t}_j)>0$ and $C\log n/\sqrt{n}\leq h/2,$ then $\|\widehat{t}_j-s_i\|_{\infty}\leq \|\widehat{t}_j-\widehat{s}_i\|_{\infty}+\|\widehat{s}_i-s_i\|_{\infty}\leq h.$ 

Hence, we have 
$$
\left|\frac{1}{n} \sum_{i=1}^{n} \frac{\sum_{j=1}^{n}[g({t}_j)-{g}(\widehat{t}_j)]  K_{H}(\widehat{s}_i, \widehat{t}_j)}{\sum_{j=1}^{n} K_{H}(\widehat{s}_i, \widehat{t}_j)}\right|\lesssim {\log n}/\sqrt{{n}}.
$$
Then following from \eqref{eq: error decompotision} and \eqref{eq: second decomp}, it is sufficient to control the following terms,
\begin{equation}
\small
\underbrace{\frac{1}{n} \sum_{i=1}^{n} {g}(\widehat{s}_i) -\int {g}({s}_i) f_{v}(v_i) dv_i}_{T_1}+
\underbrace{\frac{1}{n} \sum_{i=1}^{n} \frac{\sum_{j=1}^{n}\epsilon_j  K_{H}(\widehat{s}_i, \widehat{t}_j)}{\sum_{j=1}^{n} K_{H}(\widehat{s}_i, \widehat{t}_j)}}_{T_2}+\underbrace{\frac{1}{n} \sum_{i=1}^{n} \frac{\sum_{j=1}^{n}[g(\widehat{t}_j)-{g}(\widehat{s}_i)]  K_{H}(\widehat{s}_i, \widehat{t}_j)}{\sum_{j=1}^{n} K_{H}(\widehat{s}_i, \widehat{t}_j)}}_{T_3}.
\label{eq: key terms}
\end{equation}
We now control the three terms $T_1, T_2$ and $T_3$ separately.

\paragraph{Control of $T_1$.} The term $T_1$ is controlled by the following lemma, whose proof is presented in Section C.1.
\begin{lemma} Suppose the assumptions of Theorem \ref{thm: asf} hold, then with probability larger than $\hprob-\frac{1}{t^2},$
\begin{equation}
\left|\frac{1}{n} \sum_{i=1}^{n} {g}(\widehat{s}_i) -\int {g}({s}_i) f_{v}(v_i) dv_i\right|\lesssim {\frac{t+\log n}{\sqrt{n}}}
\label{eq: approx error}
\end{equation}
\label{lem: target approx}
\end{lemma}

\paragraph{Control of $T_2$.}
We approximate $T_2$ by 
$\frac{1}{n} \sum_{i=1}^{n} \frac{\frac{1}{n}\sum_{j=1}^{n}\epsilon_j  K_{H}({s}_i, {t}_j)}{\frac{1}{n}\sum_{j=1}^{n} K_{H}({s}_i, {t}_j)}=\frac{1}{n}\sum_{j=1}^{n}\epsilon_ja_j$ with 
\begin{equation}
a_j=\frac{1}{n} \sum_{i=1}^{n} \frac{ K_{H}({s}_i, {t}_j)}{\frac{1}{n}\sum_{j=1}^{n} K_{H}({s}_i, {t}_j)}.
\label{eq: weight definition}
\end{equation}
Then the approximation error is 
{\small
\begin{equation}
\begin{aligned}
&\frac{1}{n} \sum_{i=1}^{n} \frac{\frac{1}{n}\sum_{j=1}^{n}\epsilon_j  K_{H}(\widehat{s}_i, \widehat{t}_j)}{\frac{1}{n}\sum_{j=1}^{n} K_{H}(\widehat{s}_i, \widehat{t}_j)}-\frac{1}{n} \sum_{i=1}^{n} \frac{\frac{1}{n}\sum_{j=1}^{n}\epsilon_j  K_{H}({s}_i, {t}_j)}{\frac{1}{n}\sum_{j=1}^{n} K_{H}({s}_i, {t}_j)}\\
=&\frac{1}{n} \sum_{i=1}^{n} \frac{\frac{1}{n}\sum_{j=1}^{n}\epsilon_j [K_{H}(\widehat{s}_i, \widehat{t}_j)-K_{H}(s_i, {t}_j)]}{\frac{1}{n}\sum_{j=1}^{n} K_{H}(\widehat{s}_i, \widehat{t}_j)}+\frac{1}{n} \sum_{i=1}^{n} \frac{\frac{1}{n}\sum_{j=1}^{n}\epsilon_j  K_{H}({s}_i, {t}_j)}{\frac{1}{n}\sum_{j=1}^{n} K_{H}({s}_i, {t}_j)}\left(\frac{\frac{1}{n}\sum_{j=1}^{n} K_{H}({s}_i, {t}_j)}{\frac{1}{n}\sum_{j=1}^{n} K_{H}(\widehat{s}_i, \widehat{t}_j)}-1\right)\\
\end{aligned}
\label{eq: limiting decomp}
\end{equation}
}
The following two lemmas are needed to control $T_2$. The proofs of Lemma \ref{lem: limiting approx} and \ref{lem: limiting distribution} are presented in Section C.2 and C.3, respectively.

\begin{lemma}Suppose the assumptions of Theorem \ref{thm: asf} hold, then with probability larger than $1-n^{-C}$ for some positive constant $C>1$, for all $1\leq i\leq n,$
\begin{align}
&\frac{1}{2}f_{t}(s_i)-C\sqrt{f_{t}(s_i) \frac{\log n}{nh^3}}\leq  \frac{1}{n}\sum_{j=1}^{n} K_{H}({s}_i, {t}_j)\leq f_{t}(s_i)+C\sqrt{f_{t}(s_i) \frac{\log n}{nh^3}}\label{eq: kernel concen}\\
&\frac{1}{n}\sum_{j=1}^{n}\left|K_{H}({s}_i, {t}_j)-K_{H}(\widehat{s}_i, \widehat{t}_j)\right|\lesssim \frac{\log n}{\sqrt{n} h} \label{eq: difference concen}\\
&\left|\frac{1}{n}\sum_{j=1}^{n}\epsilon_j  K_{H}({s}_i, {t}_j)\right|\lesssim\sqrt{\frac{\log n}{nh^3}}\label{eq: kernel error}\\
&\left|\frac{1}{n}\sum_{j=1}^{n}\epsilon_j [K_{H}(\widehat{s}_i, \widehat{t}_j)-K_{H}(s_i, {t}_j)]\right|\lesssim \frac{\log n }{n^{3/4} h^2}.\label{eq: dependent concen}
\end{align}
\label{lem: limiting approx}
\end{lemma}

\begin{lemma} Suppose the assumptions of Theorem \ref{thm: asf} hold, then
\begin{equation}
\frac{\frac{1}{n}\sum_{j=1}^{n}\epsilon_ja_j}{\sqrt{\frac{1}{n^2}\sum_{j=1}^{n} {\rm Var}(\epsilon_j\mid d_j,w_j)a^2_j}}\rightarrow N(0,1)
\label{eq: variance clt}
\end{equation}
where $\epsilon_j$ is defined in \eqref{eq: reg error definition} and $a_j$ is defined in \eqref{eq: weight definition}.
With probability larger than $1-n^{-C},$
\begin{equation}
\sqrt{\frac{1}{n^2}\sum_{j=1}^{n} {\rm Var}(\epsilon_j\mid d_j,w_j)a^2_j}\asymp \frac{1}{\sqrt{nh^2}}
\label{eq: var level}
\end{equation}
\label{lem: limiting distribution}
\end{lemma}
A combination of \eqref{eq: kernel concen} and \eqref{eq: difference concen} leads to 
\begin{equation}
\frac{1}{8}f_{t}(s_i)-C\sqrt{\left|f_{t}(s_i)\right| \frac{\log n}{nh^3}}\leq \frac{1}{n}\sum_{j=1}^{n} K_{H}(\widehat{s}_i, \widehat{t}_j)\leq f_{t}(s_i)+C\sqrt{\left|f_{t}(s_i)\right| \frac{\log n}{nh^3}}.
\label{eq: kernel est concen}
\end{equation}
Together with \eqref{eq: kernel concen}, \eqref{eq: difference concen}, \eqref{eq: kernel error} and $\min_{i} f_t(s_i)\geq c_0$ for some positive constant $c_0>0,$
\begin{equation*}
\mathbf{P}\left(\left|\frac{1}{n} \sum_{i=1}^{n} \frac{\frac{1}{n}\sum_{j=1}^{n}\epsilon_j  K_{H}({s}_i, {t}_j)}{\frac{1}{n}\sum_{j=1}^{n} K_{H}({s}_i, {t}_j)}\left(\frac{\frac{1}{n}\sum_{j=1}^{n} K_{H}({s}_i, {t}_j)}{\frac{1}{n}\sum_{j=1}^{n} K_{H}(\widehat{s}_i, \widehat{t}_j)}-1\right)\right|\gtrsim \frac{(\log n)^{3/2}}{n h^{5/2}}\right)\leq n^{-C}
\end{equation*}
By \eqref{eq: kernel est concen}, \eqref{eq: dependent concen} and $\min_{i} f_t(s_i)\geq c_0,$ we have 
\begin{equation*}
\mathbf{P}\left(\left|\frac{1}{n} \sum_{i=1}^{n} \frac{\frac{1}{n}\sum_{j=1}^{n}\epsilon_j [K_{H}(\widehat{s}_i, \widehat{t}_j)-K_{H}(s_i, {t}_j)]}{\frac{1}{n}\sum_{j=1}^{n} K_{H}(\widehat{s}_i, \widehat{t}_j)}\right|\gtrsim \frac{\log n }{n^{3/4} h^2}\right)\leq n^{-C}.
\end{equation*}
Since $n h^{4}\gg (\log n)^2,$ we have 
\begin{equation*}
\sqrt{nh^2}\left|\frac{1}{n} \sum_{i=1}^{n} \frac{\frac{1}{n}\sum_{j=1}^{n}\epsilon_j  K_{H}(\widehat{s}_i, \widehat{t}_j)}{\frac{1}{n}\sum_{j=1}^{n} K_{H}(\widehat{s}_i, \widehat{t}_j)}-\frac{1}{n} \sum_{i=1}^{n} \frac{\frac{1}{n}\sum_{j=1}^{n}\epsilon_j  K_{H}({s}_i, {t}_j)}{\frac{1}{n}\sum_{j=1}^{n} K_{H}({s}_i, {t}_j)}\right|=o_{p}(1).
\end{equation*}
Together with Lemma \ref{lem: limiting distribution}, we establish that 
\begin{equation}
\frac{\frac{1}{n} \sum_{i=1}^{n} \frac{\frac{1}{n}\sum_{j=1}^{n}\epsilon_j  K_{H}(\widehat{s}_i, \widehat{t}_j)}{\frac{1}{n}\sum_{j=1}^{n} K_{H}(\widehat{s}_i, \widehat{t}_j)}}{\sqrt{\frac{1}{n^2}\sum_{j=1}^{n} {\rm Var}(\epsilon_j)a^2_j}} \rightarrow N(0,1).
\label{eq: clt mid}
\end{equation}
\paragraph{Control of $T_3$.} We decompose $T_3$ as
{\small
\begin{equation}
\begin{aligned}
\frac{1}{n} \sum_{i=1}^{n} \frac{\sum_{j=1}^{n}[{\dg}(\widehat{s}_i)]^{\intercal}(\widehat{t}_j-\widehat{s}_i)K_{H}(\widehat{s}_i, \widehat{t}_j)}{\sum_{j=1}^{n} K_{H}(\widehat{s}_i, \widehat{t}_j)}+\frac{1}{n} \sum_{i=1}^{n} \frac{\sum_{j=1}^{n}(\widehat{t}_j-\widehat{s}_i)^{\intercal}\hg(\widehat{s}_i+c_{ij}(\widehat{t}_j-\widehat{s}_i))(\widehat{t}_j-\widehat{s}_i) K_{H}(\widehat{s}_i, \widehat{t}_j)}{\sum_{j=1}^{n} K_{H}(\widehat{s}_i, \widehat{t}_j)}
\end{aligned}
\label{eq: bias separation}
\end{equation}}
for some constant $c_{ij}\in (0,1).$ 
We show that the second term of \eqref{eq: bias separation} is the higher order term, controlled as,
\begin{equation*}
\left|\frac{1}{n} \sum_{i=1}^{n} \frac{\sum_{j=1}^{n}(\widehat{t}_j-\widehat{s}_i)^{\intercal}\hg(\widehat{s}_i+c(\widehat{t}_j-\widehat{s}_i))(\widehat{t}_j-\widehat{s}_i) K_{H}(\widehat{s}_i, \widehat{t}_j)}{\sum_{j=1}^{n} K_{H}(\widehat{s}_i, \widehat{t}_j)}\right|\leq h^2 
\end{equation*}
To establish the above inequality, we apply the boundedness assumption on the hessian $\hg$ imposed in Condition 5.3 (b) and and we use the fact that, if $K_{H}(\widehat{s}_i, \widehat{t}_j)>0$ and $C\log n/\sqrt{n}\leq h/2,$ then $\|\widehat{t}_j-s_i\|_{\infty}\leq \|\widehat{t}_j-\widehat{s}_i\|_{\infty}+\|\widehat{s}_i-s_i\|_{\infty}\leq h.$

Now we control the first term of \eqref{eq: bias separation} as
{
\begin{equation}
\begin{aligned}
&\frac{1}{n} \sum_{i=1}^{n} \frac{\sum_{j=1}^{n}[{\dg}(\widehat{s}_i)]^{\intercal}(\widehat{t}_j-\widehat{s}_i)K_{H}(\widehat{s}_i, \widehat{t}_j)}{\sum_{j=1}^{n} K_{H}(\widehat{s}_i, \widehat{t}_j)}\\
&=\frac{1}{n} \sum_{i=1}^{n} \frac{\sum_{j \neq i}[{\dg}(\widehat{s}_i)]^{\intercal}(\widehat{t}_j-\widehat{s}_i)K_{H}(\widehat{s}_i, \widehat{t}_j)}{\sum_{j=1}^{n} K_{H}(\widehat{s}_i, \widehat{t}_j)}+\frac{1}{n} \sum_{i=1}^{n} \frac{[{\dg}(\widehat{s}_i)]^{\intercal}(\widehat{t}_i-\widehat{s}_i)K_{H}(\widehat{s}_i, \widehat{t}_i)}{\sum_{j=1}^{n} K_{H}(\widehat{s}_i, \widehat{t}_j)}\\
&=\frac{1}{n} \sum_{i=1}^{n} \frac{\sum_{j \neq i}[{\dg}({s}_i)]^{\intercal}({t}_j-{s}_i)K_{H}({s}_i, {t}_j)}{\sum_{j=1}^{n} K_{H}({s}_i, {t}_j)}+\frac{1}{n} \sum_{i=1}^{n} \frac{[{\dg}(\widehat{s}_i)]^{\intercal}(\widehat{t}_i-\widehat{s}_i)K_{H}(\widehat{s}_i, \widehat{t}_i)}{\sum_{j=1}^{n} K_{H}(\widehat{s}_i, \widehat{t}_j)}\\
&+\frac{1}{n} \sum_{i=1}^{n} \frac{\sum_{j \neq i}[{\dg}(\widehat{s}_i)]^{\intercal}(\widehat{t}_j-\widehat{s}_i)K_{H}(\widehat{s}_i, \widehat{t}_j)}{\sum_{j=1}^{n} K_{H}(\widehat{s}_i, \widehat{t}_j)}-\frac{1}{n} \sum_{i=1}^{n} \frac{\sum_{j \neq i}[{\dg}({s}_i)]^{\intercal}({t}_j-{s}_i)K_{H}({s}_i, {t}_j)}{\sum_{j=1}^{n} K_{H}({s}_i, {t}_j)}.
\end{aligned}
\label{eq: decorrelation}
\end{equation}
}

We introduce the following lemma to control \eqref{eq: decorrelation}, whose proof can be found in Section C.4.
\begin{lemma}
\label{lem: bias}
 Suppose that the assumptions of Theorem \ref{thm: asf} hold. Then with probability larger than $1-n^{-C}$, for some positive constant $C>0$, 
{\small
\begin{align}
&\left|\frac{1}{n} \sum_{i=1}^{n} \frac{\sum_{j \neq i}[{\dg}({s}_i)]^{\intercal}({t}_j-{s}_i)K_{H}({s}_i, {t}_j)}{\sum_{j=1}^{n} K_{H}({s}_i, {t}_j)}\right|\lesssim h^2+\sqrt{\frac{\log n}{nh}}\label{eq: main term}\\
&\left|\frac{1}{n} \sum_{i=1}^{n} \frac{[{\dg}(\widehat{s}_i)]^{\intercal}(\widehat{t}_i-\widehat{s}_i)K_{H}(\widehat{s}_i, \widehat{t}_i)}{\sum_{j=1}^{n} K_{H}(\widehat{s}_i, \widehat{t}_j)}\right|\lesssim \frac{1}{nh^2}
\label{eq: higher order 1}\\
&\left|\frac{1}{n} \sum_{i=1}^{n} \frac{\sum_{j \neq i}[{\dg}(\widehat{s}_i)]^{\intercal}(\widehat{t}_j-\widehat{s}_i)K_{H}(\widehat{s}_i, \widehat{t}_j)}{\sum_{j=1}^{n} K_{H}(\widehat{s}_i, \widehat{t}_j)}-\frac{1}{n} \sum_{i=1}^{n} \frac{\sum_{j \neq i}[{\dg}({s}_i)]^{\intercal}({t}_j-{s}_i)K_{H}({s}_i, {t}_j)}{\sum_{j=1}^{n} K_{H}({s}_i, {t}_j)}\right|
\lesssim {\frac{\log n}{\sqrt{n}}}
\label{eq: higher order 2}
\end{align}}
\end{lemma}
By applying Lemma \ref{lem: bias}, we have
\begin{equation}
\left|\frac{1}{n} \sum_{i=1}^{n} \frac{\sum_{j=1}^{n}[g(\widehat{t}_j)-{g}(\widehat{s}_i)]  K_{H}(\widehat{s}_i, \widehat{t}_j)}{\sum_{j=1}^{n} K_{H}(\widehat{s}_i, \widehat{t}_j)}\right|\lesssim h^2+\frac{1}{n h^2}+\sqrt{\frac{\log n}{nh}}
\label{eq: loose-bound}
\end{equation}
By combining \eqref{eq: approx error},  \eqref{eq: clt mid} and \eqref{eq: loose-bound}, we establish that, with probability larger than $1-\frac{1}{t^2}-n^{-C}-P(E_1)$ for some positive constant $C>0,$
\begin{equation*}
\left|\widehat{\phi}(d,w)-{\phi}^*(d,w)\right|\lesssim \frac{t}{\sqrt{nh^2}}+h^2+\frac{\log n}{\sqrt{n} }+\sqrt{\frac{\log n}{nh}}
\end{equation*}
This implies the first statement of Theorem \ref{thm: asf} under the bandwidth condition $h=n^{-\mu}$ for $0<\mu<1/4$.
Together with \eqref{eq: var level}, \eqref{eq: clt mid} and the bandwidth condition that $h=n^{-\mu}$ for $0<\mu<1/6$, we establish the asymptotic normality and the asymptotic variance level in Theorem \ref{thm: asf}.

\subsection{Proof of Corollary \ref{cor: cate}}
\label{ap-sec5}
The proof is similar to that of Theorem \ref{thm: asf}. The main extra step is to establish the asymptotic variance in the main paper. We introduce the following lemma as a modification of Lemma \ref{lem: limiting distribution} and present its proof in Section C.3.
\begin{lemma}
Suppose the assumptions of Corollary \ref{cor: cate} hold. Then
\begin{equation}
\frac{\frac{1}{n}\sum_{j=1}^{n}\epsilon_jc_j}{\sqrt{\frac{1}{n^2}\sum_{j=1}^{n} {\rm Var}(\epsilon_j\mid d_j,w_j)c^2_j}}\rightarrow N(0,1)
\label{eq: variance clt cate}
\end{equation}
where $\epsilon_j$ is defined in \eqref{eq: reg error definition} and $$c_j=\frac{1}{n} \sum_{i=1}^{n}\frac{ K_{H}({s}_i, {t}_j)}{\frac{1}{n}\sum_{j=1}^{n} K_{H}({s}_i, {t}_j)}-\frac{1}{n} \sum_{i=1}^{n}\frac{ K_{H}({r}_i, {t}_j)}{\frac{1}{n}\sum_{j=1}^{n} K_{H}({r}_i, {t}_j)}.$$ 
With probability larger than $1-n^{-C},$
\begin{equation}
\sqrt{\frac{{\rm V}_{\rm CATE}}{n}}=\sqrt{\frac{1}{n^2}\sum_{j=1}^{n} {\rm Var}(\epsilon_j\mid d_j,w_j)c^2_j}\asymp \frac{1}{\sqrt{nh^2}}
\label{eq: var level cate}
\end{equation}
\label{lem: limiting distribution cate}
\end{lemma}
Then we apply the above lemma together with the same arguments as Theorem \ref{thm: asf} to establish Corollary \ref{cor: cate}.

\bibliography{BinaryBib}{}

\begin{thebibliography}{}

\bibitem[\protect\citeauthoryear{Blundell and Powell}{Blundell and
  Powell}{2003}]{blundell2003endogeneity}
Blundell, R. and J.~L. Powell (2003).
\newblock Endogeneity in nonparametric and semiparametric regression models.
\newblock {\em Econometric society monographs\/}~{\em 36}, 312--357.

\bibitem[\protect\citeauthoryear{Blundell and Powell}{Blundell and
  Powell}{2004}]{BP04}
Blundell, R.~W. and J.~L. Powell (2004).
\newblock Endogeneity in semiparametric binary response models.
\newblock {\em The Review of Economic Studies\/}~{\em 71\/}(3), 655--679.

\bibitem[\protect\citeauthoryear{Bowden, Davey~Smith, and Burgess}{Bowden
  et~al.}{2015}]{Bowden15}
Bowden, J., G.~Davey~Smith, and S.~Burgess (2015).
\newblock Mendelian randomization with invalid instruments: effect estimation
  and bias detection through egger regression.
\newblock {\em International journal of epidemiology\/}~{\em 44\/}(2),
  512--525.

\bibitem[\protect\citeauthoryear{Bowden, Davey~Smith, Haycock, and
  Burgess}{Bowden et~al.}{2016}]{Bowden16}
Bowden, J., G.~Davey~Smith, P.~C. Haycock, and S.~Burgess (2016).
\newblock Consistent estimation in mendelian randomization with some invalid
  instruments using a weighted median estimator.
\newblock {\em Genetic epidemiology\/}~{\em 40\/}(4), 304--314.

\bibitem[\protect\citeauthoryear{Cai, Small, and Have}{Cai et~al.}{2011}]{CD11}
Cai, B., D.~S. Small, and T.~R.~T. Have (2011).
\newblock Two-stage instrumental variable methods for estimating the causal
  odds ratio: Analysis of bias.
\newblock {\em Statistics in medicine\/}~{\em 30\/}(15), 1809--1824.

\bibitem[\protect\citeauthoryear{Campbell and Mankiw}{Campbell and
  Mankiw}{1991}]{campbell1991response}
Campbell, J.~Y. and N.~G. Mankiw (1991).
\newblock The response of consumption to income: a cross-country investigation.
\newblock {\em European economic review\/}~{\em 35\/}(4), 723--756.

\bibitem[\protect\citeauthoryear{Caner, Han, and Lee}{Caner
  et~al.}{2018}]{caner2018adaptive}
Caner, M., X.~Han, and Y.~Lee (2018).
\newblock Adaptive elastic net gmm estimation with many invalid moment
  conditions: Simultaneous model and moment selection.
\newblock {\em Journal of Business \& Economic Statistics\/}~{\em 36\/}(1),
  24--46.

\bibitem[\protect\citeauthoryear{Carlson}{Carlson}{2021}]{carlson2021relaxing}
Carlson, A. (2021).
\newblock Relaxing conditional independence in an endogenous binary response
  model.
\newblock {\em Journal of Econometrics\/}.

\bibitem[\protect\citeauthoryear{Cheng and Liao}{Cheng and
  Liao}{2015}]{cheng2015select}
Cheng, X. and Z.~Liao (2015).
\newblock Select the valid and relevant moments: An information-based lasso for
  gmm with many moments.
\newblock {\em Journal of Econometrics\/}~{\em 186\/}(2), 443--464.

\bibitem[\protect\citeauthoryear{Chiaromonte, Cook, and Li}{Chiaromonte
  et~al.}{2002}]{CCL02}
Chiaromonte, F., R.~D. Cook, and B.~Li (2002).
\newblock Sufficient dimensions reduction in regressions with categorical
  predictors.
\newblock {\em The Annals of Statistics\/}~{\em 30\/}(2), 475--497.

\bibitem[\protect\citeauthoryear{Clarke and Windmeijer}{Clarke and
  Windmeijer}{2012}]{CW12}
Clarke, P.~S. and F.~Windmeijer (2012).
\newblock Instrumental variable estimators for binary outcomes.
\newblock {\em Journal of the American Statistical Association\/}~{\em
  107\/}(500), 1638--1652.

\bibitem[\protect\citeauthoryear{Conley, Hansen, and Rossi}{Conley
  et~al.}{2012}]{conley2012plausibly}
Conley, T.~G., C.~B. Hansen, and P.~E. Rossi (2012).
\newblock Plausibly exogenous.
\newblock {\em Review of Economics and Statistics\/}~{\em 94\/}(1), 260--272.

\bibitem[\protect\citeauthoryear{Cook}{Cook}{2009}]{cook98}
Cook, R.~D. (2009).
\newblock {\em Regression graphics: Ideas for studying regressions through
  graphics}, Volume 482.
\newblock John Wiley \& Sons.

\bibitem[\protect\citeauthoryear{Cook and Lee}{Cook and Lee}{1999}]{CL99}
Cook, R.~D. and H.~Lee (1999).
\newblock Dimension reduction in binary response regression.
\newblock {\em Journal of the American Statistical Association\/}~{\em
  94\/}(448), 1187--1200.

\bibitem[\protect\citeauthoryear{Cook and Li}{Cook and Li}{2002}]{Cook02}
Cook, R.~D. and B.~Li (2002).
\newblock Dimension reduction for conditional mean in regression.
\newblock {\em The Annals of Statistics\/}~{\em 30\/}(2), 455--474.

\bibitem[\protect\citeauthoryear{DiTraglia}{DiTraglia}{2016}]{ditraglia2016using}
DiTraglia, F.~J. (2016).
\newblock Using invalid instruments on purpose: Focused moment selection and
  averaging for gmm.
\newblock {\em Journal of Econometrics\/}~{\em 195\/}(2), 187--208.

\bibitem[\protect\citeauthoryear{Guo, Kang, Cai, and Small}{Guo
  et~al.}{2018}]{TSTH}
Guo, Z., H.~Kang, T.~T. Cai, and D.~S. Small (2018).
\newblock Confidence intervals for causal effects with invalid instruments by
  using two-stage hard thresholding with voting.
\newblock {\em Journal of the Royal Statistical Society: Series B (Statistical
  Methodology)\/}~{\em 80\/}(4), 793--815.

\bibitem[\protect\citeauthoryear{Guo and Small}{Guo and
  Small}{2016}]{guo2016control}
Guo, Z. and D.~S. Small (2016).
\newblock Control function instrumental variable estimation of nonlinear causal
  effect models.
\newblock {\em The Journal of Machine Learning Research\/}~{\em 17\/}(1),
  3448--3482.

\bibitem[\protect\citeauthoryear{Hahn and Ridder}{Hahn and
  Ridder}{2013}]{hahn2013asymptotic}
Hahn, J. and G.~Ridder (2013).
\newblock Asymptotic variance of semiparametric estimators with generated
  regressors.
\newblock {\em Econometrica\/}~{\em 81\/}(1), 315--340.

\bibitem[\protect\citeauthoryear{Hall and Li}{Hall and
  Li}{1993}]{hall1993almost}
Hall, P. and K.-C. Li (1993).
\newblock On almost linearity of low dimensional projections from high
  dimensional data.
\newblock {\em The annals of Statistics\/}, 867--889.

\bibitem[\protect\citeauthoryear{Hartwig, Davey~Smith, and Bowden}{Hartwig
  et~al.}{2017}]{Bowden17}
Hartwig, F.~P., G.~Davey~Smith, and J.~Bowden (2017).
\newblock Robust inference in summary data mendelian randomization via the zero
  modal pleiotropy assumption.
\newblock {\em International journal of epidemiology\/}~{\em 46\/}(6),
  1985--1998.

\bibitem[\protect\citeauthoryear{Hayfield and Racine}{Hayfield and
  Racine}{2008}]{np}
Hayfield, T. and J.~S. Racine (2008).
\newblock Nonparametric econometrics: The np package.
\newblock {\em Journal of Statistical Software\/}~{\em 27\/}(5).

\bibitem[\protect\citeauthoryear{Imbens and Rubin}{Imbens and
  Rubin}{2015}]{imbens2015causal}
Imbens, G.~W. and D.~B. Rubin (2015).
\newblock {\em Causal inference in statistics, social, and biomedical
  sciences}.
\newblock Cambridge University Press.

\bibitem[\protect\citeauthoryear{Kang, Zhang, Cai, and Small}{Kang
  et~al.}{2016}]{Kang16}
Kang, H., A.~Zhang, T.~T. Cai, and D.~S. Small (2016).
\newblock Instrumental variables estimation with some invalid instruments and
  its application to mendelian randomization.
\newblock {\em Journal of the American Statistical Association\/}~{\em
  111\/}(513), 132--144.

\bibitem[\protect\citeauthoryear{Karoui}{Karoui}{2008}]{Karoui08}
Karoui, N.~E. (2008).
\newblock {Spectrum estimation for large dimensional covariance matrices using
  random matrix theory}.
\newblock {\em The Annals of Statistics\/}~{\em 36\/}(6), 2757 -- 2790.

\bibitem[\protect\citeauthoryear{Koles{\'a}r, Chetty, Friedman, Glaeser, and
  Imbens}{Koles{\'a}r et~al.}{2015}]{kolesar2015identification}
Koles{\'a}r, M., R.~Chetty, J.~Friedman, E.~Glaeser, and G.~W. Imbens (2015).
\newblock Identification and inference with many invalid instruments.
\newblock {\em Journal of Business \& Economic Statistics\/}~{\em 33\/}(4),
  474--484.

\bibitem[\protect\citeauthoryear{Li}{Li}{1991}]{Li91}
Li, K.-C. (1991).
\newblock Sliced inverse regression for dimension reduction.
\newblock {\em Journal of the American Statistical Association\/}~{\em
  86\/}(414), 316--327.

\bibitem[\protect\citeauthoryear{Li and Guo}{Li and Guo}{2022}]{supp2}
Li, S. and Z.~Guo (2022).
\newblock Online supplement to "causal inference for nonlinear outcome models
  with possibly invalid instrumental variables".
\newblock \url{https://github.com/saili0103/SpotIV}.

\bibitem[\protect\citeauthoryear{Liao}{Liao}{2013}]{liao2013adaptive}
Liao, Z. (2013).
\newblock Adaptive gmm shrinkage estimation with consistent moment selection.
\newblock {\em Econometric Theory\/}~{\em 29\/}(5), 857--904.

\bibitem[\protect\citeauthoryear{Linton and Nielsen}{Linton and
  Nielsen}{1995}]{Linton95}
Linton, O. and J.~P. Nielsen (1995).
\newblock A kernel method of estimating structured nonparametric regression
  based on marginal integration.
\newblock {\em Biometrika\/}, 93--100.

\bibitem[\protect\citeauthoryear{Mammen, Rothe, and Schienle}{Mammen
  et~al.}{2012}]{mammen2012nonparametric}
Mammen, E., C.~Rothe, and M.~Schienle (2012).
\newblock Nonparametric regression with nonparametrically generated covariates.
\newblock {\em The Annals of Statistics\/}~{\em 40\/}(2), 1132--1170.

\bibitem[\protect\citeauthoryear{Murray}{Murray}{2006}]{murray2006avoiding}
Murray, M.~P. (2006).
\newblock Avoiding invalid instruments and coping with weak instruments.
\newblock {\em Journal of economic Perspectives\/}~{\em 20\/}(4), 111--132.

\bibitem[\protect\citeauthoryear{Newey}{Newey}{1994}]{Newey94}
Newey, W.~K. (1994).
\newblock Kernel estimation of partial means and a general variance estimator.
\newblock {\em Econometric Theory\/}~{\em 10\/}(2), 1--21.

\bibitem[\protect\citeauthoryear{Neyman}{Neyman}{1923}]{Neyman23}
Neyman, J.~S. (1923).
\newblock On the application of probability theory to agricultural experiments.
  essay on principles.
\newblock {\em Annals of Agricultural Sciences\/}~{\em 10}, 1--51.

\bibitem[\protect\citeauthoryear{Petrin and Train}{Petrin and
  Train}{2010}]{petrin2010control}
Petrin, A. and K.~Train (2010).
\newblock A control function approach to endogeneity in consumer choice models.
\newblock {\em Journal of marketing research\/}~{\em 47\/}(1), 3--13.

\bibitem[\protect\citeauthoryear{Rivers and Vuong}{Rivers and
  Vuong}{1988}]{Rivers88}
Rivers, D. and Q.~H. Vuong (1988).
\newblock Limited information estimators and exogeneity tests for simultaneous
  probit models.
\newblock {\em Journal of econometrics\/}~{\em 39\/}(3), 347--366.

\bibitem[\protect\citeauthoryear{Rosenbaum and Rubin}{Rosenbaum and
  Rubin}{1983}]{RB83}
Rosenbaum, P.~R. and D.~B. Rubin (1983).
\newblock The central role of the propensity score in observational studies for
  causal effects.
\newblock {\em Biometrika\/}~{\em 70\/}(1), 41--55.

\bibitem[\protect\citeauthoryear{Rothe}{Rothe}{2009}]{Rothe09}
Rothe, C. (2009).
\newblock Semiparametric estimation of binary response models with endogenous
  regressors.
\newblock {\em Journal of Econometrics\/}~{\em 153\/}(1), 51--64.

\bibitem[\protect\citeauthoryear{Rubin}{Rubin}{1974}]{Rubin74}
Rubin, D.~B. (1974).
\newblock Estimating causal effects of treatments in randomized and
  nonrandomized studies.
\newblock {\em Journal of educational Psychology\/}~{\em 66\/}(5), 688.

\bibitem[\protect\citeauthoryear{Tsybakov}{Tsybakov}{2008}]{tsybakov2008introduction}
Tsybakov, A.~B. (2008).
\newblock {\em Introduction to nonparametric estimation}.
\newblock Springer Science \& Business Media.

\bibitem[\protect\citeauthoryear{Vansteelandt, Bowden, Babanezhad, and
  Goetghebeur}{Vansteelandt et~al.}{2011}]{Vans11}
Vansteelandt, S., J.~Bowden, M.~Babanezhad, and E.~Goetghebeur (2011).
\newblock On instrumental variables estimation of causal odds ratios.
\newblock {\em Statistical Science\/}~{\em 26\/}(3), 403--422.

\bibitem[\protect\citeauthoryear{Wasserman}{Wasserman}{2006}]{wasserman2006all}
Wasserman, L. (2006).
\newblock {\em All of nonparametric statistics}.
\newblock Springer Science \& Business Media.

\bibitem[\protect\citeauthoryear{Windmeijer, Farbmacher, Davies, and
  Davey~Smith}{Windmeijer et~al.}{2019}]{Wind19}
Windmeijer, F., H.~Farbmacher, N.~Davies, and G.~Davey~Smith (2019).
\newblock On the use of the lasso for instrumental variables estimation with
  some invalid instruments.
\newblock {\em Journal of the American Statistical Association\/}~{\em
  114\/}(527), 1339--1350.

\bibitem[\protect\citeauthoryear{Windmeijer, Liang, Hartwig, and
  Bowden}{Windmeijer et~al.}{2019}]{Wind19b}
Windmeijer, F., X.~Liang, F.~P. Hartwig, and J.~Bowden (2019).
\newblock The confidence interval method for selecting valid instrumental
  variables.
\newblock Technical report, Department of Economics, University of Bristol, UK.

\bibitem[\protect\citeauthoryear{Wooldridge}{Wooldridge}{2010}]{wooldridge2010econometric}
Wooldridge, J.~M. (2010).
\newblock {\em Econometric analysis of cross section and panel data}.
\newblock MIT press.

\bibitem[\protect\citeauthoryear{Wooldridge}{Wooldridge}{2015}]{wooldridge2015control}
Wooldridge, J.~M. (2015).
\newblock Control function methods in applied econometrics.
\newblock {\em Journal of Human Resources\/}~{\em 50\/}(2), 420--445.

\bibitem[\protect\citeauthoryear{Xie}{Xie}{2012}]{xie2012china}
Xie, Y. (2012).
\newblock China family panel studies (2010) user’s manual.
\newblock {\em Beijing: Institute of Social Science Survey, Peking
  University\/}.

\bibitem[\protect\citeauthoryear{Zhu, Miao, and Peng}{Zhu et~al.}{2006}]{Zhu06}
Zhu, L., B.~Miao, and H.~Peng (2006).
\newblock On sliced inverse regression with high-dimensional covariates.
\newblock {\em Journal of the American Statistical Association\/}~{\em
  101\/}(474), 630--643.

\bibitem[\protect\citeauthoryear{Zhu and Fang}{Zhu and Fang}{1996}]{ZF96}
Zhu, L.-X. and K.-T. Fang (1996).
\newblock Asymptotics for kernel estimate of sliced inverse regression.
\newblock {\em The Annals of Statistics\/}~{\em 24\/}(3), 1053--1068.

\end{thebibliography}
\bibliographystyle{chicago}

\end{document}